\documentclass[12pt,letterpaper]{article}
\usepackage{natbib}

\usepackage{bbm}
\usepackage{chngcntr}
\usepackage{multirow}

\usepackage[colorlinks=true,urlcolor=blue,citecolor=blue,linkcolor=blue,bookmarks=true,bookmarksopen=false]{hyperref}

\setcitestyle{round,aysep={,},yysep={;}, citesep={;}}

\usepackage{algorithmic}
\usepackage[linesnumbered,ruled,vlined]{algorithm2e}

\usepackage[left=1in,top=1in,right=1in,bottom=1in]{geometry}

\newcommand{\exclude}[1]{}

\renewcommand{\Omega}{\varOmega}
\newcommand{\R}{\mathbbm{R}}

\newcounter{commentcounter}
\long\def\symbolfootnote[#1]#2{\begingroup%
\def\thefootnote{\fnsymbol{footnote}}\footnote[#1]{#2}\endgroup}

\usepackage{sibarticle}

\usepackage{color}

\usepackage{multirow}
\usepackage{booktabs}

\usepackage{rotating}

\usepackage{graphicx}
\usepackage{subfigure}
\usepackage{epstopdf}
\newcommand\figref{Figure~\ref}
\newcommand\tabref{Table~\ref}

\newcommand{\PP}{\mathbb{P}}

\allowdisplaybreaks

\title{Maximizing Influence in Social Networks: A Two-Stage Stochastic Programming Approach That Exploits Submodularity}
\ShortTitle{Maximizing Influence in Social Networks: A Two-Stage Stochastic Programming Approach} \ShortAuthors{Wu and K\"u\c{c}\"ukyavuz}
\NumberOfAuthors{1} \FirstAuthor{Hao-Hsiang Wu, Simge K\"u\c{c}\"ukyavuz}
\FirstAuthorAddress{Department of Integrated Systems Engineering,
The Ohio State University, Columbus, OH\\ {\tt
wu.2294@osu.edu, kucukyavuz.2@osu.edu}}

\exclude{
\SecondAuthor{}
\SecondAuthorAddress{
Department of Integrated Systems Engineering,
The Ohio State University, {\tt
wu.2294@osu.edu}
}
}

\renewcommand{\baselinestretch}{1.5}

\keywords{social networks; independent cascade; linear threshold; influence maximization; stochastic programming; submodularity}

\begin{document}

\maketitle \centerline{\today}

\begin{abstract}
We consider stochastic influence maximization problems arising in social networks. In contrast to existing studies that involve greedy approximation algorithms with a 63\% performance guarantee, our work focuses on solving the problem optimally. To this end, we introduce a new class of problems that we refer to as two-stage stochastic submodular optimization models. We propose a delayed constraint generation algorithm to find the optimal solution to this class of problems with a finite number of samples. The influence maximization problems of interest are special cases of this general problem class.  We show that the submodularity of the influence function can be exploited to develop strong  optimality cuts that are more effective than the standard optimality cuts available in the literature. Finally, we report our computational experiments with large-scale real-world datasets for two fundamental influence maximization problems, independent cascade and linear threshold, and show that our proposed algorithm outperforms the greedy algorithm. 
\end{abstract}

\renewcommand{\baselinestretch}{1.5}

\section{Introduction}

The exploding popularity of  social  networking services, such as Facebook, LinkedIn, Google+ and Twitter, has led to an increasing interest in the effective use of word-of-mouth to market products or brands to consumers. A few individuals, seen as influencers, are targeted with free merchandise, exclusive deals or new information on a product or brand.   Marketers hope that these key influencers promote the product to others in their social network through status updates, blog posts or online reviews and that this information propagates throughout the social network from peers to peers of peers until the product ``goes viral." Therefore, a key question for marketers with limited budgets and resources is to identify a small number of individuals whom to target with promotions and relevant information so as to instigate a cascade of peer influence, taking into account the network effects.



\subsection{Literature Review}

\cite{DR01} first introduce the problem of finding which customers to target to maximize the  spread of their influence in the social network. The authors propose a Markov random-field-model of the social network,  where the probability that a customer is influenced takes into account whether her connections are influenced. After building this network, the authors propose several heuristics to identify which $k$ individuals to target in a viral marketing campaign,  where $k$ is a user-defined positive integer.  
\cite{KKT03} formalize the optimization problem and introduce two fundamental  models to maximize the influence spread in a social network: the {\it independent cascade} model and the {\it linear threshold} model.  The authors  show that the optimization problems are NP-hard, assuming that there is an efficient oracle to compute the influence spread function.  This seminal work spurred a flurry of research on social networks with over 3400 citations.      \cite{WCW12} show that calculating the influence spread function is  \#P-hard under the probabilistic assumptions of \cite{KKT03}. Therefore, the independent cascade problem is  \#P-hard, and there are  two sources of difficulty. First, the calculation of the influence spread function is hard because there is an exponential number of scenarios. This difficulty is overcome by using sampling. Second, the seed selection is combinatorial in nature, and requires the evaluation of an exponential number of choices. This difficulty is overcome by seeking heuristic solutions in the literature.  We describe the results of the seminal paper by \cite{KKT03} and the subsequent developments in Section \ref{sec:greedy}.

The existing work on optimization-based methods for social network analysis focus on various aspects other than influence maximization \cite[see the review by][]{XABP09}. The first class of problems studied is that of identifying the influential nodes of a network with respect to the nodes' centrality and connectivity. As an example of this class of problems, \cite{ACEP09} propose an integer programming formulation for the problem of identifying $k$ nodes whose removal from a {\it deterministic} social network causes maximum fragmentation (disconnected components). The second class is that of clustering the nodes of a {\it deterministic} social network to identify the cohesive subgroups of the network. For example, \cite{BBH11} and \cite{EVB16} utilize  optimization models to identify clique relaxations. Third, game-theoretic approaches are used to study various aspects of social networks, such as modeling competitive marketing strategies of two firms to maximize their market shares \cite[see, e.g.,][]{BOY16}. 

In contrast to these models, we focus on the stochastic influence maximization problems and propose a two-stage stochastic programming method. In addition, by utilizing the submodularity of the second stage value (objective) function, we develop effective decomposition algorithms. Two-stage stochastic programming is a versatile modeling tool for decision-making under uncertainty. In the first stage, a set of decision needs to be made when some parameters are random. In the second stage, after the uncertain parameters are revealed, a second set of (recourse) decisions are made so that the expected total cost is minimized.  We refer the reader to
 \cite{BL97} and \cite{SDR09} for an overview of stochastic (linear) programming. 
 To the best of our knowledge, \cite{SD14} provide the only  study besides ours that uses a stochastic programming approach to solve a problem in social networks. In this paper, the authors consider the problem of protecting some arcs of a social network (subject to a limited budget) so that the damage caused by the spread of rumors from their sources to a set of targeted nodes is minimized. 

\subsection{Our contributions}

Despite the ubiquity of social networks, there has been a paucity of research in finding provably optimal solutions to the two fundamental problems of maximizing influence in social networks (independent cascade and general threshold). The algorithms studied to date are approximation algorithms with a worst-case guarantee within 63\% optimal \cite[][and references therein]{KKT15}. The proposed heuristics are tested on real social networks and compared to other simple heuristics. However, their practical performance has not been tested against the optimal solution due to the hardness of the problem and the unavailability of an algorithm that can find the optimal solution for large-scale instances of the problem. To fill this gap, we introduce a new class of problems that we refer to as two-stage stochastic submodular optimization models. We propose a delayed constraint generation algorithm to find the optimal solution to this class of problems with a finite number of samples. The proposed delayed constraint generation algorithm  exploits the submodularity of the second-stage value function. The influence maximization problems of interest are special cases of this general problem class.   Utilizing the special structure of the influence function, we give an explicit characterization of  the cut coefficients of the submodular inequalities, and  identify conditions under which they are facet-defining for the full master problem that is solved by delayed constraint generation. This leads to a  more efficient implementation of the proposed  algorithm than is available from a textbook implementation of available algorithms for this class of problems \cite[]{B62,VW69,NW81}. In addition, we give the complete linear description of the master problem for $k=1$. We illustrate our proposed algorithm on  the
classical {\it independent cascade} and {\it linear threshold} problems \cite[]{KKT03}. In our computational study, we show that our algorithm yields solutions with 36\% higher optimality guarantees, much faster than the greedy heuristic in most of the large-scale real-world  instances.

We  note that while we demonstrate our algorithms on the independent cascade and linear threshold models, our approach is more generally applicable to many other variants of the influence maximization problem studied previously in the literature. Furthermore, beyond social networks, there are other applications of identifying a few key nodes in complex networks for which our models are applicable. For example, \cite{OS04} consider the problem of locating costly sensors on the crucial junctures of the water distribution network to ensure water quality and safety by the early detection and prevention of outbreaks.  The models could also be useful in the development of immunization strategies in epidemic models \cite[see, e.g.][]{MKCBH04},  and prevention of cascading failures in power systems \cite[see, e.g.,][]{HBS09}.  Furthermore, it also applies to more general stochastic optimization problems that have submodular second-stage value functions. For example, recently \cite{CF14} consider a {\it deterministic} hub location problem, and prove that the routing costs in the objective function  are submodular. Using this observation, the authors employ the delayed constraint generation algorithm of \cite{NW81}  to solve the optimization problem more effectively than the existing models for this problem. Our proposed algorithm can be used to solve a {\it stochastic} extension of the hub location problem, where in the first stage, the hub locations are determined, and in the second stage, after the revelation of uncertain demand of multiple commodities, the optimal routing decisions are made. Hence, the general two-stage stochastic submodular optimization model and method  that we introduce in Section \ref{sec:benders} has a potential broader impact beyond social networks. 

\subsection{Outline} In Section \ref{sec:greedy}, we formally introduce the influence maximization problem and review the greedy algorithm of \cite{KKT03}. In Section \ref{sec:benders}, we define a  general two-stage stochastic submodular optimization model, and describe a delayed constraint generation  algorithm that exploits the submodularity of the second-stage value function. We show that for $k=1$, solving a linear program with a simple set of submodular optimality cuts and the cardinality restriction on the seed set guarantees an integer optimal solution.   In Section \ref{sec:iclt}, we consider the two fundamental influence maximization problems as defined by \cite{KKT03}, namely {\it independent cascade} and {\it linear threshold}. We show that for these special cases of the two-stage stochastic submodular optimization problems, we can obtain an explicit form of the submodular optimality cuts and identify conditions under which they are facet defining. In Section \ref{sec:comp}, we report our computational experience with large-scale real-world datasets, which show the efficacy of the proposed approach in finding  optimal solutions as compared to the greedy algorithm. We share our conclusions and future work in Section \ref{sec:conc}. 

\section{Greedy Algorithm of \cite{KKT03}} \label{sec:greedy}

In this section, we describe the modeling assumptions of \cite{KKT03}, and overview the  greedy hill-climbing algorithm proposed by these authors. Suppose that we are given a social network $G=(V,A)$, where $|V|=n, |A|=m$. The vertices represent the individuals, and an arc $(i,j)\in A$ represents a potential influence relationship between individuals $i$ and $j$.   Our goal is to select a subset of {\it seed} nodes, $X\subset V$, with $|X|\le k<n$ to activate initially, so that the expected number of people influenced by $X$  (denoted by $\sigma(X)$) is maximized, where $k$ is a given integer. (Note that the original problem statement is to select exactly $k$ nodes to activate. However, for the relaxation that seeks $|X|\le k$ seed nodes that maximize influence, there exists a solution for which the inequality holds at equality.)  The influence propagation  is assumed to be {\it progressive}, in other words, once a node is activated it remains active.

\cite{KKT03} show that for various influence maximization problems, the influence function $\sigma(X)$ is nonnegative, monotone and submodular. Therefore, the influence maximization problem involves the maximization of a submodular function. The authors show that this problem is NP-hard even if there is an efficient oracle to compute the influence spread function.  However, using the results of \cite{CFN77} and \cite{NWF78} that the greedy method gives a $(1-\frac{1}{e})$-approximation algorithm for maximizing a nonnegative monotone submodular function, where $e$ is the base of the natural logarithm, \cite{KKT03} establish that the greedy hill-climbing algorithm solves the influence maximization problem with a constant (0.63) guarantee, assuming that the function  $\sigma(X)$ can be calculated efficiently.  
Recognizing the computational difficulty of calculating $\sigma(X)$ exactly, which involves taking the expectation of the influence function with respect to a finite (but  exponential) number of scenarios, \cite{KKT03} propose Monte-Carlo sampling, which provides a  subset of equiprobable scenarios, $\Lambda$,  of moderate size. Letting $\sigma_\omega$ denote the influence function for scenario $\omega\in \Lambda$, we get $\sigma(X)=\frac{1}{|\Lambda|}\sum_{\omega\in \Lambda}\sigma_\omega(X)$. The basic greedy approximation algorithm of  \cite{KKT03}  is given in Algorithm \ref{alg:greedy}.

Subsequently,  \cite{WCW12} formally show that calculating $\sigma(X)$ is  \#P-hard under the assumption of independent arc probabilities $\pi_{ij}, (i,j)\in A$. Therefore, \cite{KKT15} propose a modification where an arbitrarily good approximation of  $\sigma(X)$ is obtained in polynomial time by sampling from the true distribution. In particular,  \cite{KKT15} show that for a sample size  of   $\Omega\left(\frac{n^2}{\varepsilon^2}\ln (1/\alpha)\right)$, the average number of activated nodes over the sample is a $(1\pm \varepsilon)$-approximation to $\sigma(X)$, with probability at least $1-\alpha$. 

\begin{algorithm}\label{alg:greedy}
 \SetAlgoLined
Start with $X=\emptyset$ and a sample set of scenarios $\Lambda$\;
 \While{$|X|\le k$}
 {
 	For each node $i\in V\setminus X$, use the sample $\Lambda$ to approximate $\sigma(X\cup\{i\})$\; 
Add  node $i$ with the largest estimate for $\sigma(X\cup\{i\})$ to $X$\;
}
Output the set  of seed nodes, $X$.
\caption{Greedy Approximation Algorithm of \cite{KKT03}.}
\end{algorithm}

Further algorithmic improvements to the greedy heuristic are given in the literature \cite[see][for an overview]{KKT15,CLC13}.  Most notably, \cite{BBCL14} give a randomized algorithm for finding a $(1-1/e-\epsilon)$-approximate seed sets in $O((m+n)\epsilon^{-3}\log  n)$ time for any precision parameter  $\epsilon>0$. Note that this run time is independent of the number of seeds $k$. The authors show that the running time is close to the lower bound of $\Omega(m+n)$ on the time required to obtain  a constant factor randomized  approximation algorithm. The proposed randomized algorithm has a success probability of 0.6, and failure is detectable. Therefore, the authors suggest repeated runs if failure is detected to improve the probability of success.  


\section{A  Two-Stage Stochastic Submodular Optimization Model and Method}\label{sec:benders}

In this section, we define a general two-stage stochastic submodular optimization model and  outline a delayed constraint generation algorithm for its solution. Then, in Section \ref{sec:iclt}, we describe how this general model and method is applicable to the influence maximization problems of interest. 

Let $(\Lambda, \mathcal F, \PP)$ be a finite probability space, where the probability of an elementary event $\omega\in \Lambda$ is  $p_\omega:=\PP(\omega)$.   Consider a general two-stage stochastic binary  program
\begin{subequations}\label{model:2sso}
\begin{align}
\max~~& c^\top x +\sum_{\omega\in \Lambda} p_\omega \sigma_\omega(x)\\
\text{s.t.}~~& x\in \mathcal X  \\
& x\in\{0,1\}^n,
\end{align}
\end{subequations}
where $c\in \R^n$ is a given objective vector, the set $\mathcal X$ represents the constraints on the first-stage variables $x$ and $\sigma_\omega(x)$ is the objective function of the second-stage problem for scenario $\omega\in \Lambda$ solved as a function of first-stage decisions given by 
\begin{subequations}\label{model:gensub}
\begin{align}
\sigma_\omega(x):=\max~~& q^\top y\\
\text{s.t.}~~& y\in \mathcal Y(x,\omega).
\end{align}
\end{subequations}
Here $q$ is an objective vector of conformable dimension, $y$ is the vector of second-stage decisions, and $\mathcal Y(x,\omega)$ defines the set of feasible second-stage decisions for a given first-stage vector $x$, and the realization of the uncertain outcomes given by the scenario $\omega\in \Lambda$. 
We assume that $\sigma_\omega(x):\{0,1\}^n\rightarrow \R$ is known to be a submodular function for each $\omega\in \Lambda$, and refer to the optimization problem \eqref{model:2sso} as a {\it two-stage stochastic submodular optimization model}.  
 It is well-known from the property of submodular functions that if  $\sigma_\omega(x), \omega\in \Omega$ is submodular, then so is the second-stage value function $\sigma(x)= \sum_{\omega\in \Lambda} p_\omega \sigma_\omega(x)$, which is a nonnegative  (convex) combination of submodular functions. Furthermore, we assume that $\mathcal Y(x,\omega)$ is a non-empty set for each $x\in  \mathcal X, \omega\in \Lambda$, a property known as {\it relatively complete recourse} in stochastic programming. 

Next we overview a delayed constraint generation approach  to solve the two-stage program \eqref{model:2sso}.  The generic master problem at an iteration is formulated as 
\begin{subequations}\label{master}
\begin{align}
\max~~& c^\top x +\sum_{\omega\in \Lambda} p_\omega \theta_\omega\\
\text{s.t.}~~&x\in \mathcal X  \\
&(x,\theta)\in \mathcal C  \label{eq:optcuts}, 
\end{align}
\end{subequations}
where $\theta$ is a $|\Lambda|$-dimensional vector of variables  $\theta_\omega$  representing the second-stage objective function approximation for scenario $\omega$,  constraints \eqref{eq:optcuts} represents the so-called {\it optimality cuts}  generated until this iteration. The set of inequalities in $\mathcal C$ provide a piecewise linear approximation of the second stage value function, which is iteratively refined through the addition of  the optimality cuts.  (We will describe different forms of these inequalities in the following discussion.) Let $(\bar x,\bar \theta)$ be the optimal solution to the master problem at the current iteration. Then for all $\omega\in \Lambda$ we solve the subproblems \eqref{model:gensub} to obtain $\sigma_\omega (\bar x)$. We add  valid optimality cuts  to $\mathcal C$ if $\bar \theta_\omega>\sigma_\omega(\bar x)$ for any $\omega\in \Lambda$, otherwise we  deduce that the current solution $\bar x$  is optimal. The generic version of the delayed constraint generation algorithm is given in Algorithm \ref{alg:benders}. In this algorithm, $\varepsilon$ is a user-defined optimality tolerance.   The particular implementation of Algorithm \ref{alg:benders} depends on  the method with which subproblems are solved to obtain $\sigma_\omega(\bar x)$ (in line \ref{alg:step-solve} of Algorithm \ref{alg:benders}), and the form of the optimality cuts added to the master problem (in line \ref{alg:step-cut} of Algorithm \ref{alg:benders}). In this section, we explore the possibility of utilizing the submodularity of the second-stage value function in a two-stage stochastic programming problem.  We discuss a natural alternative in Appendix \ref{sec:app}, which we use as a benchmark.

\begin{algorithm}\label{alg:benders}
 \SetAlgoLined
Start with $\mathcal C=\{0\le  \theta_\omega\le n, \omega\in \Lambda \}$. Let LB=0 and UB=$n$\;
 \While{${\text{UB}-LB}\le \varepsilon$}
 {
 Solve the  master problem \eqref{master} and obtain $(\bar x, \bar \theta)$. Let UB be the upper bound obtained from the optimal objective value of the master problem\;
\For{$\omega\in \Lambda$}
{\do
Solve Subproblem \eqref{model:gensub} to obtain   $\sigma_\omega(\bar x)$   \label{alg:step-solve}\; 
\If{$\bar \theta_\omega>\sigma_\omega (\bar x)$}
{Add an optimality cut to $\mathcal C$\; \label{alg:step-cut}
}
}
Let $\sigma(\bar x)=\sum_{\omega\in \Lambda}p_\omega \sigma_\omega (\bar x)$.
 \If{LB $<\sigma(\bar x)$}
 {Let LB  $\gets\sigma(\bar x)$, and let $\hat x\gets \bar x$ be the incumbent solution
 }
}
Output the set  of seed nodes $X=\{i\in V: \hat x_i=1\}$.
\caption{Delayed Constraint Generation Algorithm.}
\end{algorithm}

 \cite{NW81} give submodular inequalities  to describe the maximum of a submodular set function \cite[see also][]{NW88}. 
Consider the polyhedra $\mathcal S_\omega=\{(\theta_\omega,x)\in \R\times\{0,1\}^n:\theta_\omega\le \sigma_\omega(S) +\sum_{j\in V\setminus S} \rho^\omega_j(S) x_j, \forall S\subseteq V\}$,  and $ \mathcal S'_\omega=\{(\theta_\omega,x)\in \R\times\{0,1\}^n:\theta_\omega\le \sigma_\omega(S) -\sum_{j\in S} \rho^\omega_j(V\setminus\{j\})(1-x_j) +\sum_{j\in V\setminus S} \rho^\omega_j(S) x_j, \forall S\subseteq V\}$ for $\omega \in \Lambda$, where $\rho_j^\omega(S)=\sigma_\omega(S\cup\{j\})-\sigma_\omega(S)$ is the marginal contribution of adding $j\in V\setminus S$ to the set $S$.  

\begin{theorem} \cite[cf.][] {NW81} \label{thm:submodeq} 
For a submodular and nondecreasing set function $\sigma_\omega:2^n\rightarrow \R$, $\bar X$, with a characteristic vector $\bar x$, is an optimal solution to $\max_{S\subseteq V: |S|\le k} \{\sigma_\omega(S)\}$, if and only if $(\theta_\omega, \bar x)$ is an optimal solution to $\{\max \  \theta_\omega : \sum_{j\in V} x_j\le k, {(\theta_\omega,x)\in \mathcal S_\omega} \}$.  Similarly for a submodular and nonmonotone set function $\sigma_\omega:2^n\rightarrow \R$, $\bar X$, with a characteristic vector $\bar x$, is an optimal solution to $\max_{S\subseteq V: |S|\le k} \{\sigma_\omega(S)\}$, if and only if $(\theta_\omega, \bar x)$ is an optimal solution to $\{\max \  \theta_\omega : \sum_{j\in V} x_j\le k, {(\theta_\omega,x)\in  \mathcal  S'_\omega} \}$.
\end{theorem}



Therefore, we can adapt the delayed constraint generation algorithm of \cite{NW88}  given for deterministic submodular maximization problems to {\it}  two-stage stochastic submodular optimization problems.   
The proposed method takes the form of Algorithm \ref{alg:benders}. For a given first stage solution, $\bar x$, which is a characteristic vector of the set $\bar X$, and scenario $\omega\in \Lambda$, we use  the optimality cut 
\begin{equation}\label{eq:gen-nd}
\theta_\omega\le \sigma_\omega(\bar X) +\sum_{j\in V\setminus \bar X} \rho^\omega_j(\bar X) x_j,
\end{equation}
if the second-stage value function $\sigma_\omega(x)$ is nondecreasing and submodular. 
 If the second-stage value function $\sigma_\omega(x)$ is nonmonotone and submodular, then we use the optimality cut given by 
 the inequality
\begin{equation}\label{eq:gen-nm}
\theta_\omega\le \sigma_\omega(\bar X) -\sum_{j\in \bar X} \rho^\omega_j(V\setminus\{j\})(1-x_j) +\sum_{j\in V\setminus \bar X} \rho^\omega_j(\bar X) x_j.
\end{equation}
We refer the reader to \cite{NW81} for validity of inequalities \eqref{eq:gen-nd}-\eqref{eq:gen-nm}.
\begin{corollary} \label{prop:finconv}
Algorithm \ref{alg:benders} with optimality cuts \eqref{eq:gen-nd} and \eqref{eq:gen-nm} converges to an optimal solution in finitely many iterations   for a  two-stage stochastic program with binary first-stage decisions, $x\in\{0,1\}^{|V|}$ for which the second-stage value function, $\sigma_\omega(x), \omega\in \Lambda$, ($|\Lambda|$ finite) is submodular nondecreasing and submodular nonmonotone, respectively.
\end{corollary}

\begin{proof}
The result follows from the fact that the number of feasible first stage solutions is finite, and from Theorem \ref{thm:submodeq}. 
\end{proof}

Note that Algorithm \ref{alg:benders} is  generally applicable to two-stage stochastic programs with binary first-stage decisions, $x\in\{0,1\}^{n}$, where the second-stage value function, $\sigma_\omega(x)$ is submodular for all $\omega\in \Lambda$.  There is very limited reporting on the  computational performance of this algorithm even for deterministic submodular maximization problems  for which the method was originally developed \cite[see][for computational results on quadratic cost partition and hub location problems, respectively]{LNW96,CF14}. To the best of our knowledge, our work is  the first adaptation and testing of this algorithm to stochastic optimization. While the submodular inequalities \eqref{eq:gen-nd}-\eqref{eq:gen-nm} are implicit in that they require the calculation of $\rho^\omega_j(\cdot)$ terms, in Section \ref{sec:iclt}, we give an explicit form of the submodular optimality cuts for influence maximization problems of interest. This allows us to characterize conditions under which the optimality cuts are strong, and to improve the performance of a textbook implementation of the algorithm of \cite{NW81}.


\exclude{
Inequalities \eqref{eq:gen-nd} and \eqref{eq:gen-nm} require the calculation of the marginals $\rho^\omega_j$ for deterministic subproblems.  
For example,  \cite{W89} give implicit submodular inequalities for capacitated fixed-charge network flow problems, where the marginals are determined by maximum flow problems in an appropriate network. For structured graphs, \cite{ATK15} give efficient algorithms to calculate the marginals.
}

\subsection*{Convex Hull for a  Special Case}
Next, we consider the special case of cardinality-constrained first-stage problem \eqref{model:2sso}, i.e., $\mathcal X:=\{x\in\{0,1\}^n: \sum_{j\in V}x_j\le k\}$, when $k=1$. It is easy to see  that in this case, the greedy algorithm is optimal. Note also that for fixed $k$, the problem is polynomially solvable (with respect to the input size of number of nodes, arcs and scenarios), because it involves evaluating $O(n^k)$  possible functions $\sigma_\omega(X), \omega\in \Lambda$.  
Observe that, without loss of generality, we can assume that $p_\omega >0$ for all $\omega\in \Lambda$ (otherwise, we can ignore scenario $\omega$), and that $\sigma_\omega(\emptyset)=0$ (otherwise, we can add a constant to the influence function). Furthermore, because  $\sigma_\omega(\cdot)$ is submodular $\rho_j^\omega(\emptyset)\ge \rho_j^\omega(S)$ for any $S\subseteq V, S\ne \emptyset$ and $j\in V\setminus S$. As a result, if $\rho_j^\omega(\emptyset)<0$, then $x_j=0$ in any optimal solution. Therefore, without loss of generality, we can assume that $\rho^\omega_j(\emptyset)\ge 0$ for all $j\in V, \omega\in \Lambda$.

\begin{proposition}\label{prop:k=1}
	For  submodular  functions $\sigma_\omega(x), \omega\in \Lambda$ with $\rho^\omega_j(\emptyset)>0$ for all $j\in V, \omega\in \Lambda$, and  $\mathcal X:=\{x\in\{0,1\}^n: \sum_{j\in V}x_j\le 1\}$,  adding the submodular optimality cut \eqref{eq:gen-nd} with $\bar X = \emptyset$ to the linear programming (LP) relaxation of the master problem \eqref{master} for each scenario is sufficient to give the (integer) optimal solution $x^*$.   
\end{proposition}

\begin{proof}
First, note that for  $\bar X=\emptyset$, inequalities  \eqref{eq:gen-nd} and \eqref{eq:gen-nm} are equivalent. Under the given assumptions,  in an optimal solution $x\ne 0$, the right-hand side  of \eqref{eq:gen-nd}  is  positive for each $\omega\in \Lambda$. Therefore, the decision variables $\theta_\omega > 0, \omega\in \Lambda$,  are  basic variables at an extreme point  optimal solution of the LP relaxation of the master problem \eqref{master}. This gives us $|\Lambda|$ basic variables,  and the number of constraints is $|\Lambda|+1$. Hence, only one decision variable  $x_j$ for some $j\in V$ can be basic, and it is equal to $1$ (due to  constraint \eqref{eq:card}), and $\theta_\omega=\rho_j^\omega(\emptyset)=\sigma_\omega(\{j\})$. Furthermore, this is the optimal solution to the master problem for the case $k=1$.	 
\end{proof}

\section{Triggering Set Technique and the Live-arc Graph Representation} \label{sec:iclt}

In this section,  we specify how the general algorithm we propose for two-stage stochastic programs with submodular second-stage value functions applies to the influence maximization problems of interest. 
 \cite{KKT03} observe that even though the stochastic diffusion process of influence spread is dynamic, because the decisions of whom to activate do not influence the probability of an individual influencing another, we may envision the process to be static and ignore the time aspect. In other words, we can generate sample paths (scenarios) of likely events for each arc, a priori. 
As a result, the decision-making  process considered by \cite{KKT03} may be viewed as a two-stage stochastic program. In the first stage, the nodes to be activated are determined. The uncertainty, represented by a finite collection of scenarios,  $\Lambda$, is revealed with respect to how  the influence spreads in the network. For each scenario $\omega\in \Lambda$, with associated probability $p_\omega$, the influence spread given the initial seed set $X$ is calculated as  $ \sigma_\omega(X)$. As a result, the expected total influence spread of the initial seed set $X$  is given by  $\sigma(X)=\sum_{\omega\in \Lambda} p_\omega \sigma_\omega(X)$. Let $x\in \{0,1\}^n$ be the characteristic vector of $X\subset V$. Where appropriate, we use $\sigma(x)$ interchangeably with $\sigma(X)$. 

As observed by \cite{KKT03}, the influence function $\sigma_\omega(X)$ is submodular and monotone (nondecreasing) for various influence maximization problems. Then  the two-stage stochastic programming formulation of the classical influence maximization problem is given by \eqref{model:2sso} where 
$c_j=0$ for al $j\in V$ and  the set $\mathcal X$ defines the cardinality constraint on the number of seed nodes given by
\begin{equation}
\sum_{j\in V} x_j \le k,  \label{eq:card}
\end{equation}
for a given $0<k<|V|$. Therefore, Algorithm \ref{alg:benders} can be used to solve the influence maximization problem. Furthermore, note that the influence functions of interest in this paper  satisfy the assumption  $\rho^\omega_j(\emptyset)>0$ for all $j\in V, \omega\in \Lambda$, because influencing only node $j$ contributes at least one node (itself) to the influence function. In addition, the first-stage problem is cardinality-constrained. Hence Proposition \ref{prop:k=1} applies to the influence functions considered in this paper. 

To model the stochastic diffusion process  and calculate the influence spread function, 
\cite{KKT03} introduce a technique that generates a finite set, $\Lambda$, of sample paths (scenarios) by tossing biased coins. The coin tosses reveal, a priori, which influence arcs are active (live). A live-arc $(i,j)$ indicates that if node $i$ is influenced during the influence propagation process, then  node $j$ is influenced by it. 
 For each scenario $\omega\in \Lambda$, with a probability of occurrence  $p_\omega$, a so-called   {\it live-arc graph}  $G_\omega=(V,A_\omega)$ is constructed, where $A_\omega$ is the set of {\it live} arcs  under scenario $\omega$.  The influence spread under  scenario $\omega\in \Lambda$, denoted as $\sigma_\omega(X)$, is then calculated as  the number of vertices reachable from $X$ in $G_\omega$. Hence, the expected influence spread function is given by  $\sigma(X)=\sum_{\omega\in \Lambda}p_\omega\sigma_\omega(X)$. 
This is referred to as the ``triggering model" or the ``triggering set technique" by \cite{KKT15}. 
The authors show the equivalence of the stochastic diffusion process of two fundamental influence maximization problems to the live-arc graph model  with respect to the final active set. In addition, \cite{KKT03} show that the influence spread in a live-arc graph representable problem is monotone and submodular under the given assumptions. As a result, our stochastic programming method applies to such problems.  Next we describe the two fundamental influence maximization problems that are live-arc representable. 

\begin{description}

\item[{\bf Independent Cascade Model:}] 
In the independent cascade model of \cite{KKT03}, it is assumed that each arc $(i,j)\in A$ of the social network $G=(V,A)$ has an associated probability of success, $\pi_{ij}$. In other words, with probability $\pi_{ij}$ individual $i$ will be successful at influencing individual $j$.  We say that an arc $(i,j)$ is {\it active} or {\it live} in this case. We generate a sample path (scenario) by tossing biased coins (with probability of $\pi_{ij}$ for each arc $(i,j)\in A$) to determine whether the arc is active/live to construct the live-arc graph. Because each arc influence probability is independent, and does not depend on which nodes are influenced,  \cite{KKT03} show that the influence maximization problem is equivalent to maximizing the expected influence function in the live-arc graph model.

\item[{\bf Linear Threshold Model:}] 
In  the linear threshold model of \cite{KKT03},  each arc $(i,j)$ in the social network $G=(V,A)$ has deterministic weight $0\le w_{ij}\le 1$, such that for all nodes $j\in V$, $\sum_{i:(i,j)\in A}w_{ij}\le 1$. In addition, each node $j\in V$ selects a threshold $\nu_j$ {\it uniformly at random}. A node is activated if sum of the weights of its {\it active} neighbors is above the thresholds, i.e., $\sum_{i:(i,j)\in A} w_{ij}x_i\ge \nu_j$.  Given the set  of initial  seed nodes, $\bar X$, the activated nodes in the set $U$ at time $t$  influence their unactivated neighbor $j$ at time $t+1$ if $\sum_{u\in U} w_{uj} \geq \nu_j $.  
	\cite{KKT03} show that the linear threshold model also has an equivalent live-arc graph representation, where every node has at most one incoming live arc. 
Each node $j \in V$   selects at most one incoming live arc $(i,j)$ with probability $w_{ij}$, or it selects no arc with probability $1-\sum_{i:(i,j)\in A}w_{ij}$. Given the seed set $\bar X$, \cite{KKT03} prove the following two are equivalent:
	\begin{enumerate}	
		\item The distribution of active nodes computed by executing the linear threshold model with starting seed set $\bar X$, and  
		\item the distribution of nodes reachable from $\bar X$ in the live-arc graph representation of the linear threshold model defined above. 	
	\end{enumerate}

\end{description}

Next, we demonstrate how the proposed  algorithm (Algorithm \ref{alg:benders}) can be applied to influence maximization problems that have a live-arc graph representation. Subsequently, we  give extensions where the proposed algorithm applies to models which are not live-arc graph representable. In such models, the form of the cuts change, but as long as the influence spread function is submodular, the proposed algorithm applies.  

\subsection{Exploiting the Submodularity of the Second-Stage Value Function for Live-Arc Graph Models}

Utilizing Theorem \ref{thm:submodeq}, we give an explicit description of the submodular inequalities for the influence maximization  problems that have live-arc graph representations. We say that a node $j$ is reachable from a set of nodes $S$, in scenario $\omega\in \Lambda$, if there exists a node $i\in S$ such that there is a directed path from $i$ to $j$ in the graph  $G_\omega=(V, A_\omega)$. It is well known that reachability can be checked in linear time with respect to the number of arcs using depth- or breadth-first search.  For  $S\subseteq V$ and $\omega\in \Lambda$, let $R(S)$ be the set of nodes  reachable from the nodes in $S$ not including the nodes in $S$, and let  $\bar R(S)$ be the set of nodes not reachable from the nodes in $S$ in the graph $G_\omega=(V, A_\omega)$. 
\begin{proposition}
For $S\subseteq V$ and $\omega\in \Lambda$ the inequality
\begin{equation}\label{eq:optcut-submod}
\theta_\omega\le \sigma_\omega(S)+\sum_{j\in \bar R(S)} r_j^\omega(S) x_j,
\end{equation}
is a valid optimality cut for the master problem \eqref{master}, where   $r_j^\omega(S)$ is the number of nodes reachable from $j\in \bar R(S)$ (including $j$) that are not reachable from any node in $S$ in $G_\omega$. 
\end{proposition}
\begin{proof}
From Theorem \ref{thm:submodeq}, we know that $\theta_\omega\le \sigma_\omega(S) +\sum_{j\in V\setminus S} \rho^\omega_j(S) x_j$ is a valid inequality. Note that $\bar R(S)\subseteq V\setminus S$ and for $j\in \bar R(S)$, we have $\rho^\omega_j(S)=r_j^\omega(S)$, in other words, the marginal contribution of adding $j\in \bar R(S)$ to $S$ is precisely $r_j^\omega(S)$. Furthermore, for any node  $j\in R(S)$, the marginal contribution of adding $j$ to $S$ is zero, because $j$ is already reachable from at least one node in $S$. This completes the proof. 
\end{proof}

We
 refer to the cuts in the form of \eqref{eq:optcut-submod} as {\it submodular optimality cuts}.
Next we give conditions under which inequalities  \eqref{eq:optcut-submod}  are facet defining for conv($\mathcal S_\omega$). For $i\in V$, let indeg($i$) and outdeg($i$) denote the in-degree and out-degree of node $i$, respectively. Let $T:=\{i\in V: \text{indeg}(i)=0\}$, we refer to the nodes in $T$ as {\it root nodes}. For $i\in V\setminus T$, let $P_i$ be the set of root nodes such that $i$ is reachable from the nodes in this set, i.e., $P_i:=\{j\in T: i\in R(\{j\})\}$. Finally, let $L:=\{i\in V: \text{indeg}(i)>0,\text{outdeg}(i)=0\}$ denote the set of {\it leaf nodes} that have no outgoing arcs.

\begin{proposition} \label{prop:facet}
For $S\subseteq V$ and $\omega\in \Lambda$ the submodular inequality  \eqref{eq:optcut-submod}  is facet defining for conv($\mathcal S_\omega$) {\it only if} the following conditions hold
\begin{enumerate}
\item if $i \in S$, then $i\not \in T$, 
\item there exists $T'\subseteq T$ with $|T'|<k$ such that $S\subseteq R(T')$.
\end{enumerate}
These conditions are also sufficient 
\begin{enumerate}
\item if $S=\emptyset$ (for any $k\ge 1$), or
\item if $|S|=1$ for $k\ge 2$.
\end{enumerate}
\end{proposition}

\begin{proof}
\noindent{\it Necessity.} First, note that the submodular inequality  \eqref{eq:optcut-submod}  for a set $S$ is equivalent to that for the set $S\cup R(S)=:\hat R(S)$, because $\sigma_\omega(S)=\sigma_\omega(\hat R(S))$, $\bar R(S)=\bar R(\hat R(S))$,  $r_j^\omega(S)=r_j^\omega(\hat R(S))$ for all $j\in \bar R(S)$ and $\rho_j^\omega(S)=0$ for $j\in R(S)$. Therefore, without loss of generality, we assume that for all non-leaf nodes $i\in S\setminus L$, we have $R(\{i\})\subseteq S$ (Assumption {\bf A1}). 
\begin{enumerate}
\item Suppose, for contradiction, that there exists $i \in S \cap T$. Now consider the submodular inequality  \eqref{eq:optcut-submod} for the set $S'=S\setminus\{i\}$ given by 
\begin{equation}\label{eq:dom-T}
\theta_\omega\le \sigma_\omega(S')+\sum_{j\in \bar R(S')} r_j^\omega(S') x_j=\sigma_\omega(S)-1+x_i+\sum_{j\in \bar R(S)} r_j^\omega(S) x_j,
\end{equation}
which follows because the set of all descendants of $i$, $R(\{i\})$ is contained in $S$ by Assumption {\bf A1}, so removing $i$ reduces the influence function by exactly 1 (recall that, by the contradictory assumption $i\in T$, hence its in-degree is 0 and it is not influenced by any other node in the graph), and  the set of nodes not reachable from $S'$  is given by $\bar R(S')=\bar R(S)\cup\{i\}$, and hence the coefficients $r_j^\omega(S') =r_j^\omega(S)$ for $j\in \bar R(S)$, and $r_i^\omega(S') =1$.  Because $x_i\le 1$, inequality \eqref{eq:dom-T} dominates the submodular inequality  \eqref{eq:optcut-submod}  for this choice of $S$. Hence, the  submodular inequality for a set $S$ such that  there exists $i \in S \cap T$  is not facet defining for conv($\mathcal S_\omega$).
\item Suppose, for contradiction, that there does not exist $T'\subseteq T$ with $|T'|<k$ such that $S\subseteq R(T')$. In other words, the minimum cardinality of root nodes $T'\subseteq T$ such that $S\subseteq R(T')$ is greater than or equal to $k$. In this case, consider the set $\hat S:=\{i\in S: \nexists j\in S \text{ with } i\in R(\{j\})\}$, in other words, $\hat S$ is the set of nodes in the graph induced by $S$ that have no incoming arcs from other nodes in $S$. Note that from condition (i),  we know that $\hat S\cap T=\emptyset$. Then, by the contradictory assumption,  
there exist at least $k$ nodes, say nodes $1,\ldots,k\in \hat S$ such that $P_i \cap P_j=\emptyset$ for all pairs $i,j\in\{1,\ldots,k\}, i\ne j$. Now consider the submodular inequality  \eqref{eq:optcut-submod} for the set $S'=S\setminus\{1,\ldots,k\}$ given by 
\begin{equation}\label{eq:dom-hatS}
\theta_\omega\le \sigma_\omega(S')+\sum_{j\in \bar R(S')} r_j^\omega(S') x_j=\sigma_\omega(S)-k+\sum_{j\in \bar R(S)} r_j^\omega(S) x_j + \sum_{i=1}^k \sum_{j\in \hat R(P_i)\setminus R(\{i\})}x_j ,
\end{equation}
 which follows because the set of all descendants of $i\in\{1,\ldots,k\}$, $R(\{i\})$, is contained in $S$ by Assumption {\bf A1}, so removing nodes $i=1,\ldots,k$ reduces the influence function by exactly $k$, and  the set of nodes not reachable from $S'$  is given by $\bar R(S')=\bar R(S)\cup\{1,\ldots,k\}$. In addition, the coefficients $r_j^\omega(S') =r_j^\omega(S)$ for $j\in \bar R(S)$ such that $j\not \in \cup_{i=1}^k\left (\hat R(P_i)\setminus R(\{i\})\right)$, $r_j^\omega(S') =r_j^\omega(S)+1$ for $j\in \bar R(S)$ such that $j \in \cup_{i=1}^k\left (\hat R(P_i)\setminus R(\{i\})\right)$,  and $r_i^\omega(S') =1$ for $i=1,\ldots,k$. 
Because $ \sum_{i=1}^k \sum_{j\in\hat R(P_i)\setminus R(\{i\})}x_j\le \sum_{j\in V} x_j \le k$, inequality \eqref{eq:dom-hatS} dominates the submodular inequality  \eqref{eq:optcut-submod}  for this choice of $S$. Hence, there must exist $T'\subseteq T$ with $|T'|<k$ such that $S\subseteq R(T')$ for the  submodular inequality \eqref{eq:optcut-submod} to be facet defining for conv($\mathcal S_\omega$). 

\end{enumerate}

\noindent{\it Sufficiency.}  First, note that for $\omega\in \Lambda$, dim$(\mathcal S_\omega)=n+1$. Let $\mathbf e_i$ be a unit vector of dimension $n$ whose $i$th component is 1, and other components are zero.
\begin{enumerate}
\item  Note that when $S=\emptyset$, the necessity conditions are trivially satisfied.  Consider the $n+1$ affinely independent  points: $(\theta_\omega,x)^0=\mathbf 0$, and $(\theta_\omega,x)^i=(\sigma_\omega(\{i\}), \mathbf e_i)$, for $i\in V$. These  points  are on the face defined by the inequality  \eqref{eq:optcut-submod} for $S=\emptyset$. Hence  inequality  \eqref{eq:optcut-submod} for $S=\emptyset$ is facet-defining for conv($\mathcal S_\omega$).

\item Note that for $|S|=1$, the necessity conditions imply that $S:=\{j\}$ for some $j\in L$. Consider the $n+1$ affinely independent  points: $(\theta_\omega,x)^0=(\sigma_\omega(\{j\}), \mathbf 0)$; $(\theta_\omega,x)^j=(\sigma_\omega(\{j\}), \mathbf e_j$  and $(\theta_\omega,x)^i=(\sigma_\omega(\{i,j\}), \mathbf e_j+\mathbf e_i)$, for $i\in V\setminus\{j\}$. The last set of points is feasible because we have $k\ge 2$ in this case.  These  points  are on the face defined by the inequality  \eqref{eq:optcut-submod} for $S=\{j\}$. Hence  inequality  \eqref{eq:optcut-submod} for $S=\{j\}$ is facet-defining for conv($\mathcal S_\omega$).

\end{enumerate}

\end{proof}

Note that during the course of the algorithm, if a submodular inequality \eqref{eq:optcut-submod}  corresponding to the seed set $S$ does not satisfy the necessary conditions given in Proposition \ref{prop:facet}, then a stronger inequality can be constructed using the arguments in the proof of the proposition. 

\subsection{Facet conditions at work}  

From Proposition \ref{prop:facet} we see that inequalities \eqref{eq:optcut-submod} with $S=\emptyset$ are facets of conv($\mathcal S_\omega$) for any $k\ge 1$. We will also see their importance in  our computational study. Similarly, inequalities \eqref{eq:optcut-submod} with $|S|=1$ are facets of conv($\mathcal S_\omega$) for any $k\ge 2$. We note that more conditions are necessary for the inequalities \eqref{eq:optcut-submod} with $|S|=2$  to be facets of conv($\mathcal S_\omega$). We illustrate this in the next example. 

\begin{example}\label{ex:toy}
Consider the network in Figure \ref{Fig:1024} for a given scenario $\omega\in \Lambda$ and let $k=2$. From Proposition \ref{prop:facet}, inequalities \eqref{eq:optcut-submod} with $S=\emptyset$, and inequalities \eqref{eq:optcut-submod} with $S=\{j\}$, for $j=4,\ldots,9$ are facet-defining for conv($\mathcal S_\omega$). Inequalities \eqref{eq:optcut-submod} with $S=\{7,8\}$ or   $S=\{5,6\}$ are facets of conv($\mathcal S_\omega$);  each of these sets satisfy the necessary facet conditions in Proposition \ref{prop:facet}, which for these choices  of $S$  also turn out to be sufficient. 
However, the sets  $S=\{7,9\}$ or   $S=\{4,5\}$ satisfy the necessary facet conditions in Proposition \ref{prop:facet}, but they do not lead to facet-defining inequalities for $k=2$.  
Finally, $S=\{4,7\}$ violates the necessity condition (ii) of Proposition \ref{prop:facet} (the minimum number of root nodes that can influence 4 and 7 is $2=k$) and is not a facet.  
\begin{figure}[htbp]
	\centering	
	\includegraphics[width=8.5cm]{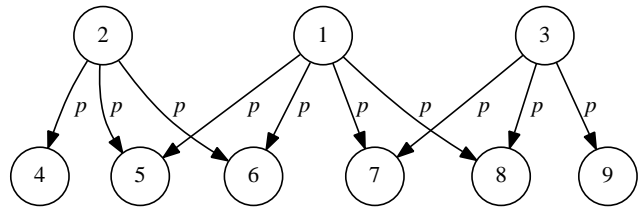}
	\caption{Network with 9 nodes and 10 arcs with equal influence probabilities $p$.}
	\label{Fig:1024} 
\end{figure}

\end{example}

It is important to note that in the direct adaptation of the  delayed constraint generation algorithm proposed by \cite{NW88} to our problem,  for a given solution $\bar x$ to the current master problem, one would use the submodular inequalities \eqref{eq:optcut-submod}, where we let $S=\{i\in V: \bar x_i=1\}=:\bar X$.   From Proposition \ref{prop:finconv}, we have that Algorithm \ref{alg:benders} with optimality cuts \eqref{eq:optcut-submod} with $S=\bar X$ converges to an optimal solution in finitely many iterations   for  two-stage stochastic submodular maximization problems. However, note that at any iteration, the solution to the master problem, $\bar x$ will be such that  $\sum_{j\in V} \bar x_j=k$. Therefore, all submodular optimality cuts \eqref{eq:optcut-submod} added will have $|S|=k$, and no facets with $S=\emptyset$  will be added for $k\ge 1$.  
This may  lead to slow convergence of the delayed constraint generation algorithm with submodular optimality cuts for $S=\bar X$, which we illustrate next. 

\renewcommand{\theexample}{\ref{ex:toy}}
\begin{example} (Continued.)
Consider the network in Figure \ref{Fig:1024}, and suppose that $|\Lambda|=1$, hence we consider a deterministic problem.  Adding the submodular optimality cut  \eqref{eq:optcut-submod} for $S=\emptyset$: $\theta_1\le 5x_1+4x_2+4x_3+\sum_{j=4}^9 x_j$ to the cardinality constraint, and solving the linear programming relaxation of the master problem yields the integer optimal solution, $\bar x_1=1, \theta_1=5$ at the first iteration (from Proposition \ref{prop:k=1}). In contrast, solving the  master problem without any optimality cuts (i.e., with just the cardinality constraint) may lead to an initial solution of $\bar x_2=1$. Then the following set of submodular cuts are added in that order during the course of the algorithm of \cite{NW88}: 
\begin{align}
\theta_1&\le 4+3x_1+4x_3+\sum_{j=7}^9 x_j & (S=\bar X=\{2\}), \label{eq:ex-toy1}\\
\theta_1&\le 4+3x_1+4x_2+\sum_{j=4}^6 x_j & (S=\bar X=\{3\}),\\
\theta_1&\le 5+2x_2+2x_3+x_4+x_9 & (S=\bar X=\{1\}). \label{eq:ex-toy3}
\end{align}
Furthermore, none of the  inequalities \eqref{eq:ex-toy1}-\eqref{eq:ex-toy3} are facet-defining. Solving the LP relaxation of the master problem with the optimality cuts \eqref{eq:ex-toy1}-\eqref{eq:ex-toy3} leads to a fractional solution: $\bar x_2=\bar x_3=0.5$. This small example highlights that a textbook implementation of the algorithm by \cite{NW88} may lead to slow convergence because  the algorithm (1) may explore, in the worst case,  $O{{n}\choose{k}}$ many locally optimal solutions before finding an optimal solution, and (2) may require long solution times for each master problem, because the optimality cuts given by $S=\bar X$ may not be facet-defining and hence  lead to  weak LP relaxations. 
\end{example}

These observations enable us to devise a more efficient implementation of Algorithm \ref{alg:benders}, which we  report in Section \ref{sec:comp}.

\exclude{the associated characteristic vector $\bar X\cap T\ne \emptyset$ (otherwise we can strictly increase the objective value by at least one). Therefore, the  submodular optimality cuts \eqref{eq:optcut-submod} with $S=\bar X$ used in such an implementation are not facet-defining, from necessary facet condition (i) of Proposition \ref{prop:facet}.   However, our facet proof immediately yields a stronger inequality to use than the one given by $S=\bar X$.  In addition, because}

\exclude{ 

In addition, in most sparse networks the optimal solution often corresponds to $k$ connected components and there may be an exponential number of alternative optima. In such cases, the textbook implementation of \cite{NW88}  explores all alternative solutions before it declares optimality.  Hence, if a tight upper bound on the influence function is added to the master problem for such instances, then the convergence of the algorithm can be drastically improved, as we show next.

\begin{proposition} \label{prop:nooverlap}
Suppose that at any iteration of  Algorithm \ref{alg:benders}, for a given $k\ge 2$, the optimal solution to the current master problem, $\bar x$, is such that the graph induced by the nodes in  $\bar X\cup R(\bar X)$ has $k$ connected components in $G_\omega$ for all $\omega\in \Lambda$. Then  inequality \eqref{eq:optcut-submod} with $S=\emptyset$ is guaranteed to terminate the algorithm and declare $\bar x$ as the optimal solution. 
\end{proposition}

\begin{proof}
Note that each solution to the master problem provides a lower bound LB$=\sigma(\bar x)$ on the optimal objective function value. Furthermore, if $\bar X\cup R(\bar X)$ have $k$ connected components in $G_\omega$ for all $\omega\in \Lambda$, then $\theta_\omega\le\sum_{j\in \bar X}\rho_j^\omega(\emptyset)$. In addition, if $\bar X$ has $k$ connected components, then we must have $\bar X\subseteq T$ and $ \sum_{j\in \bar X}\rho_j^\omega(\emptyset) = \sigma_{\omega}(\bar x)$. Hence, inequality  \eqref{eq:optcut-submod} with $S=\emptyset$ provides an upper bound on the optimal objective function value given by  $\sum_{\omega\in \Lambda}p_\omega \theta_\omega \le \sigma(\bar x)$. Therefore, $\bar x$ must be optimal.  
\end{proof}
}

\exclude{
\subsection{Sampling issues}

Note that while sampling issues are not considered in the present paper, there is a rich body of work on sampling for stochastic programs. In particular, \cite{KSH02} show that for discrete stochastic optimization problems,  for a sample size of $\Omega\left(\frac{3V}{\epsilon^2}\ln \left(\frac{|\mathcal X|}{\alpha}\right)\right)$, the probability that the optimal solution to the sampled problem is an $\epsilon$-optimal solution to the original problem is at least $1-\alpha$.  Here $V$ is a parameter bounding the maximum variance of the difference between an optimal objective value and the objective value of a non-$\epsilon$-optimal solution. Observe that because $|\mathcal X| $ has $O(n^k)$ elements in our case, the  sample size estimate grows in the order of $k \ln n$. 
Furthermore, after a candidate solution is found, one can utilize multiple replications method for assessing solution quality to obtain asymptotically valid confidence intervals for the optimality gap \cite[see][for a survey of sampling methods for stochastic programming]{HB14}. 

}

\subsection{Extensions}

In this section, we give various extensions of the influence maximization problems that can be solved using our proposed methods. 

\subsubsection{Extensions to live-arc graph models}

Observe that while we demonstrate our general algorithm  on the independent cascade and linear threshold models, our proposed model and method is applicable to many extensions of the social network problems studied in the literature. 
For example, an extension considered in the literature is to replace the cardinality constraint on the number of nodes selected with a knapsack  constraint representing a marketing budget where each node has a different cost to market. This model also admits an adapted and more involved 0.63-factor greedy approximation algorithm \cite[see,][]{KMN99,S04}. In fact, our model is flexible enough to allow any constraints in $\mathcal X$ so long as the master problem can be solved with an optimization solver, while the greedy approximation algorithm needs careful adjustment and analysis for each additional constraint. Similarly,  the time-constrained influence spread problem studied in \cite{CLZ12} and \cite{LCXZ12} can also be solved using our method. In this problem, there is an additional constraint that the number of time periods it takes to influence a node should be no more than a given parameter $\tau$.  The resulting influence spread function is monotone and submodular, hence we can use  inequalities \eqref{eq:gen-nd} as the submodular optimality cuts.  
Furthermore, we can efficiently calculate the coefficients $\rho^\omega_j(\bar X) $ by solving, with breadth-first search, a modified reachability problem limiting the number of hops from the seed set $\bar X$ to any other node by $\tau$.

\subsubsection{General  cascade and general    threshold models}

In the {\it general cascade model}, every node $j\in V$ has an {\it activation function} $p'_j(i,S)\in[0,1]$ for $S\subseteq \{(k,j)\in A\}=:N^{in}(j)$ and $i\in  N^{in}(j)\setminus S$. The activation function represents the probability that node $j$ is influenced by node $i$ given that the nodes in $S$ failed to activate node $i$. The independent cascade model is a special case, where $p'_j(i,S)=\pi_{ij}$, independent of $S$. 

In the {\it general threshold  model},  every node $j\in V$ has an {\it threshold function} $f_j(S)$ for $S\subseteq N^{in}(j)$, where $f_j(\cdot)$ is monotone and $f_j(\emptyset)=0$. As before, every node $j$ selects a threshold $\nu_j$ uniformly at random in the range $[0,1]$. Then, a node $j$ is activated if for a given active set $S$, $f_j(S\cap  N^{in}(j))\ge \nu_j$. The linear threshold model is a special case, where $f_j(S)=\sum_{i\in S}w_{ij}$. 

\cite{KKT03} show that general cascade model is equivalent to the general threshold model with an appropriate selection of activation and threshold functions. This is not true for the independent cascade and the linear threshold models \cite[see Example 2.14 in][]{CLC13}. Furthermore, the influence spread function is no longer submodular. However, if $f_j(S)$ is submodular for all $j\in V$, then the influence spread is submodular (first conjectured by \cite{KKT03} and later proven by \cite{MR07,MR10}). Therefore, the greedy hill climbing algorithm is a 0.63-approximation algorithm for this case as well. Algorithm \ref{alg:benders} is applicable  in the submodular threshold functions case, where the optimality cuts take the more general form \eqref{eq:gen-nd} or \eqref{eq:gen-nm} depending on the monotonicity of the function $f$.

\section{Computational Experiments}\label{sec:comp}

In this section we summarize our experience with solving the influence maximization problem using the   delayed constraint generation method  (DCG) with various optimality cuts as given in Algorithm \ref{alg:benders}, and the greedy hill-climbing algorithm (Greedy) of \cite{KKT03} as given in Algorithm \ref{alg:greedy}. 
The algorithms are implemented in C++ with IBM ILOG CPLEX 12.6 Optimizer. All experiments were executed on a Windows Server 2012 R2 with an Intel Xeon E5-2630 2.40 GHz CPU, 32 GB DRAM and x64 based processor.   
In our implementation of Algorithm  \ref{alg:benders}, we set the parameter $\varepsilon=0$.  
 For the master problem of the decomposition algorithm, the relative MIP gap tolerance of CPLEX was set to 1\%, so a feasible solution which has an optimality gap of 1\% is considered optimal.


\subsection{Small-Scale Network} First, we study the quality of the solutions produced by  DCG and Greedy on a small-scale network for which we can enumerate all possible outcomes of the random process. In these experiments, we are able to capture the random process precisely, and no information is lost through sampling from the true distribution. 
 An illustrative network  is given in   \figref{Fig:1024} with 9 nodes, 10 directed arcs and  independent influence probability $\pi_{ij}=p$ for all $(i,j)\in A$. Our goal is to select $k=2$ seed nodes, so that the objective value, which is the expected number of nodes influenced by the seed nodes, is maximized. We generate all possible influence  scenarios (a total of $2^{10}=1024$ scenarios). Note that under the assumption that each influence is independent of the others, the probability of  scenario $\omega$,  which has $\ell \leq 10$ live arcs, is given by $p_\omega = \left( 1-p \right)^{  10-\ell  } p^\ell$. 

The solution of DCG and Greedy methods on  1024 scenarios with various values of $p = 0.1, 0.2,\ldots, 1$ is shown in \tabref{Table:1024}. When $p \leq 0.5$, both algorithms have the same objective value. For  $ 0.6\le p \le 1$, Greedy  selects  node 1 as the seed in the first iteration of Algorithm \ref{alg:greedy} (line 4 of Algotihm \ref{alg:greedy}) and selects either node 2 or 3 as the seed in the second iteration. However, DCG selects nodes 2 and 3 as the seed nodes, and  provides a better objective value than Greedy (up to 12.5\% improvement). So while Greedy does better than its worst-case bound (63\%), it is within 12.5\% of optimality.  

Next, instead of generating all 1024 scenarios, we employed Monte-Carlo sampling, and independently sampled  different number of scenarios $|\Lambda| = 10, 50$ and 100  according to different $p$ values, and let $p_\omega = 1/|\Lambda|$. We summarize the results of this experiment in  \tabref{Table:sample1024}. For eight out of 15 cases, DCG  has a higher objective value than Greedy, and in all other cases Greedy attains the optimal objective value (mostly for small influence probabilities $p=0.1, 0.3$). We also observe that the objective value for the instances with a larger number of scenarios is generally  closer to the objective value with all 1024 scenarios (except for $p=0.1$ and $|\Lambda| =100$). Note that Greedy is a 0.63-approximation algorithm even for the sampled problem, which assumes that the true distribution is given by the  scenarios in $\Lambda$, whereas DCG provides the optimal solution to the sampled problem. 


\exclude{
	, and the small number of scenarios still gives us a good result in this problem. 
	In this paper, we focus on showing that our methods scales well with respect to the number of scenarios representing the uncertainty. The reason why different influence probability $p$ or different number of sampling scenarios causes the different objective value of two algorithms raises another issue, and it will be discussed as the future work.
}

\begin{table}[htb]
	\caption{Expected influence  obtained from two algorithms for  the small-scale network with 1024 scenarios.}
	\label{Table:1024}
	\begin{center}
		\begin{tabular}{|c|c|c|c|c|c|c|c|c|c|c|} \hline
			& \multicolumn{10}{|c|}{Objective values with different $p$}  \\
			\cline{2-11} %
			Algorithm & $p=1.0$ & $p=0.9$ & $p=0.8$ & $p=0.7$ & $p=0.6$ & $p=0.5$ & $p=0.4$ & $p=0.3$ & $p=0.2$ & $p=0.1$\\
			\hline
			DCG & 8 & 7.4 & 6.8 & 6.2 & 5.6 & 5 & 4.48 & 3.92 & 3.32 & 2.68 \\
			Greedy & 7 & 6.68 & 6.32 & 5.92 & 5.48 & 5 & 4.48 & 3.92 & 3.32 & 2.68 		 
			\\\hline%
		\end{tabular}
		
	\end{center}
\end{table}

\begin{table}[htb]
	\caption{Expected influence  obtained from two algorithms for the small-scale network  with $|\Lambda|$ scenarios.}
	\label{Table:sample1024}
	\begin{center}
		\begin{tabular}{|cc|c|c|c|c|c|} \hline
			& & \multicolumn{5}{|c|}{Objective values for different $p$}  \\
			\cline{3-7} %
			$|\Lambda|$ & Algorithm & $p=0.9$ & $p=0.6$ & $p=0.5$ & $p=0.3$ & $p=0.1$ \\
			\hline
			10 & DCG & 7.1 & 5.2 & 4.7 & 3.4 & 2.6 \\
			10 & Greedy & 6.8 & 5.1 & 4.6 & 3.4 & 2.6 	 \\\hline%
			50 & DCG & 7.4 & 5.84 & 5.18 & 3.98 & 2.68 \\
			50 & Greedy & 6.82 & 5.66 & 5.06 & 3.98 & 2.68 \\\hline%
			100 & DCG & 7.38 & 5.61 & 5.04 & 3.96 & 2.76 \\
			100 & Greedy & 6.69 & 5.52 & 5.04 & 3.96 & 2.76 	 \\\hline%
		\end{tabular}
		
	\end{center}
\end{table}

\subsection{Large-Scale Network with Real-World Datasets} \label{sec:hepnet}

To evaluate the efficiency of DCG and Greedy on large networks, we conduct computational experiments on  four real-world datasets with different categories and scales. These datasets are summarized below and in  \tabref{Table:data_sum}:
\begin{description}
\item [{\bf UCI-message}] is the dataset for the online student community network at the University of California, Irvine  \cite[]{UCI}. The 1,899 nodes represent the students, and the 59,835 directed arcs between two nodes indicate that one student sent a message to the other. 
\item[{\bf P2P02}] is the dataset for the Gnutella peer-to-peer file sharing network from August 2002 \cite[]{p2p1,p2p2}. The 10,876 nodes represent the hosts, and the 39,994 undirected edges denote the connections between two hosts. 
\item[{\bf Phy-HEP}] is the dataset for the academic collaboration network in the ``high energy physics theory" (HEPT) section of the e-print arXiv (\url{www.arxiv.org}) \cite[]{KKT03,CWY09,KKT15}. The 15,233 nodes represent the authors, and the 58,891 undirected edges represent the co-authorship between each pair of authors in the ``high energy physics theory" papers from 1991 to 2003. Note that this is the original dataset considered in   \cite{KKT03}, and it is commonly used as a benchmark in comparing various algorithms for maximizing influence in social networks. 
\item[{\bf Email-Enron}] is the dataset for  the email communication network of  the Enron Corporation \cite[]{email1,email2}.  It is posted to the public by the Federal Energy Regulatory Commission during the investigation. The 36,692 nodes represent the different email addresses, and the 183,831 directed arcs denote one address sent a mail to the other. 
\end{description}
Note that if the  graph in the original datasets contain  undirected edges  between $i$ and $j$, then we  construct  a directed graph with two directed arcs from $i$ to $j$ and $j$ to $i$. We follow the data generation scheme of   \cite{KKT03} in constructing live-arcs graphs for these instances  \cite[i.e., the probabilities of influence on each arc for the independent cascade model, and the arc weights and node thresholds for the linear threshold model follow from][]{KKT03}.  

In our experiments in this subsection, we compare three algorithms:  {\it Greedy} is the greedy hill-climbing  algorithm (Algorithm \ref{alg:greedy});  {\it DCG-SubIneqs} is Algorithm \ref{alg:benders}  using submodular optimality cuts \eqref{eq:optcut-submod} for $S=\bar X$, where $\bar X$ is the optimal solution given by the master problem in the current iteration;  and  {\it DCG-SubWarmup} adds the submodular inequalities \eqref{eq:optcut-submod} for $S = \emptyset $ for each scenario, before the execution of DCG-SubIneqs as a warm-start.  Proposition \ref{prop:k=1} shows that the submodular optimality cut \eqref{eq:optcut-submod} with $S = \emptyset$, which is referred to as EmptySetCut, is sufficient to find the optimal solution for $k=1$ (note that $\rho^\omega_j(\emptyset)\ge 1$ for all $j\in V, \omega\in \Lambda$, hence the assumptions of the proposition are satisfied). Our goal in implementing DCG-SubWarmup is to test  if EmptySetCut is also useful for  $k>1$. To verify this, we add EmptySetCuts for all scenarios to the master problem before executing the DCG algorithm, and solving the initial master problem. Note that the total computation time in our experiments includes the generation time of all EmptySetCuts, and the total number of user cuts also includes the number of EmptySetCuts.   We also implemented Algorithm  \ref{alg:benders} using alternative optimality cuts  \cite[adapted and strengthened versions of integer-L-shaped cuts of][ referred to as Benders-LC, and described in Appendix \ref{sec:app}]{LL93}; 
however, the running time of Benders-LC is extremely slow. Therefore, we only report our results with DCG-SubIneqs, DCG-SubWarmup and Greedy, and discuss the inefficiency of Benders-LC   in Appendix \ref{sec:lc-comp}.

\begin{table}[htb]
	\caption{The summary of real world datasets.}
	\label{Table:data_sum}
	\begin{center}
		\begin{tabular}{|c|c|c|c|c|} \hline
			&  \multicolumn{4}{|c|}{Dataset}  \\
			\cline{2-5} %
			& UCI-message & P2P02 & Phy-HEP & Email-Enron \\
			\hline
			Network Category  & Online-Message  & File-Shearing  & Collaboration & Communication 	 \\
			Nodes  	          & 1,899                    & 10,876          & 15,233 & 36,692 	 \\
			Edges             & 59,835                 & 39,994          & 58,891 & 183,831 \\			
			Format            & Directed            & Undirected     & Undirected & Directed 	 \\\hline%
		\end{tabular}
		
	\end{center}
\end{table}

\subsubsection{Independent Cascade Model} \label{sec:ICexp}

For the independent cascade model, we assign uniform influence probability $\pi_{ij}=p=0.1$ independently to each arc $(i,j)$ in the network as was done in  \cite{KKT03}. Note that  \cite{KKT03}  consider the dataset {\bf Phy-HEP} with influence probabilities $\pi_{ij}\in\{0.01, 0.1\}$ for each arc $(i,j)$ in the network. However, we observe that for $\pi_{ij}=0.01$,  the total number of  live arcs is very small, resulting in  sparse live-arc graphs with  a large number of  singletons. For example, the expected number of live arcs in our largest dataset {\bf Email-Enron} is $183,831*0.01 = 1,838$ with $p=0.01$. However, the number of nodes of  {\bf Email-Enron}  is 36,692, resulting in over 30,000 singletons in the network. Therefore,  we  focus on the more interesting case of $\pi_{ij}=0.1$ in our experiments in this section. 

We generate  $|\Lambda| =  100, 200, 300$  and 500 scenarios  to find $k=1$ to 5 seed nodes that maximize influence.  Tables  \ref{Table:IC-table1}-\ref{Table:IC-table2} summarize our experiments with the algorithms DCG-SubIneqs, DCG-SubWarmup and Greedy for the independent cascade model. 
Column ``$k$" denotes the number of seed nodes to be selected. Column ``Cuts(\#)" reports the total number of submodular inequalities \eqref{eq:optcut-submod} added to the master problem of DCG-SubIneqs, and column  ``Time(s)" reports the solution time in seconds. We do not report the objective values in these experiments, because we are able to prove that despite its worst-case performance guarantee of 63\%,  Greedy is within the optimality tolerance  for these instances. In \cite{KKT03} Greedy is tested empirically against other heuristics such as choosing the nodes with $k$ highest degrees in the graph $G$, because it is said that an optimal solution is not available. Therefore, our computational experiments also provide an empirical test on the greedy heuristic when the optimal solution is available (to the sampled problem) due to our proposed method.  

\begin{sidewaystable}[htp]
	\caption{Independent Cascade Model for UCI-message and P2P02} 
	\label{Table:IC-table1}
	\centering 
	\begin{tabular}{ |p{1cm}p{1cm}||p{1.5cm}p{1.5cm}|p{1.5cm}p{1.5cm}|p{1.5cm}||p{1.5cm}p{1.5cm}|p{1.5cm}p{1.5cm}|p{1.5cm}|  }
		\hline
		
		& &
		\multicolumn{5}{c||}{UCI-message} &
		\multicolumn{5}{c|}{P2P02}  \\
		
		\cline{3-12}%
		& 
		& \multicolumn{2}{c|}{ {  DCG-SubIneqs} } & \multicolumn{2}{c|}{ {  DCG-SubWarmup} } & \multicolumn{1}{c||}{ {Greedy} } 
		& \multicolumn{2}{c|}{ {  DCG-SubIneqs} } & \multicolumn{2}{c|}{ {  DCG-SubWarmup} } & \multicolumn{1}{|c|}{ {Greedy} }			
		\\

		{$k$} & { $|\Lambda|$}  
		& {Time(s) }& { Cuts(\#) } & {Time(s) }  & {  Cuts(\#) }&  { Time(s)} 		
		& {Time(s) }& { Cuts(\#) } & {Time(s) }  & { Cuts(\#) }&  { Time(s)} 				
		\\	\hline
		
		
		1 & 100  & 50& 118 & 45 & 100 & 31 & 460& 200 & 227&100&211 \\
		2 & 100  & 75& 200 & 72 & 200 & 40 & 466& 200 & 505&212&364\\
		3 & 100  & 82& 200 & 78 & 201 & 50 & 463& 216 & 514&216&520\\
		4 & 100  & 76& 200 & 71 & 200 & 59 & 721& 245 & 591&244&682\\
		5 & 100  & 81& 200 & 73 & 200 & 69 & 572& 232 & 547&232&823 \\\hline%
		
		1 & 200  & 55& 234 & 43 & 200 & 59 & 538& 400 &  261&200&324\\
		2 & 200  & 90& 399 & 85 & 400 & 77 & 526& 400 &  514&400&558\\
		3 & 200  & 91& 400 & 88 & 400 & 96 & 553& 412 &  552&411&794\\
		4 & 200  & 95& 400 & 84 & 400 & 118 & 569& 433 &  552&428&1014\\
		5 & 200  & 89& 400 & 87 & 400 & 138 & 785& 451 &  590&451&1227 \\\hline%
		
		1 & 300  & 186& 339 & 164 & 300 & 99 & 234& 600 &  136&300&1029\\
		2 & 300  & 356& 596 & 260 & 600 & 133 & 240& 600 &  246&600&1852\\
		3 & 300  & 274& 600 & 264 & 600 & 168 & 246& 600 &  251&600&2652\\
		4 & 300  & 268& 600 & 282 & 600 & 201 & 324& 658 &  295&643&3452\\
		5 & 300  & 273& 600 & 272 & 600 & 232  & 310& 678 &  302&658&4369\\ \hline%

		1 & 500  & 161 & 571 & 99 & 500 & 154 & 529 &	1000 & 271 & 500 & 2508\\
		2 & 500  & 237 & 994 & 201 & 1000 & 202 & 522 &	1000 & 515 & 1000 & 4796\\
		3 & 500  & 221 & 999 & 207 & 1000 & 252 & 489	 & 1000 & 522 & 1000 & 6922\\		
		4 & 500  & 219 & 1000 & 203 & 1000 & 299 & 609	& 1075 & 721 & 1193 & 9018\\
		5 & 500  & 211 & 1000 & 200	& 1000 & 347 & 608 & 1131 & 620 & 1129 & 11043\\ \hline%
		
		\multicolumn{2}{|c||}{\bf Average}  & 159.5& 502.5 &143.9&495.05& 141.2	 & 488.2& 576.55 &436.6&525.85& 2707.9 	 \\ \hline

	\end{tabular}
\end{sidewaystable}

\begin{sidewaystable}[htp]
	\caption{Independent Cascade Model for Phy-HEP and Email-Enron} 
	\label{Table:IC-table2}
	\centering 
	\begin{tabular}{ |p{1cm}p{1cm}||p{1.5cm}p{1.5cm}|p{1.5cm}p{1.5cm}|p{1.5cm}||p{1.5cm}p{1.5cm}|p{1.5cm}p{1.5cm}|p{1.5cm}|  }
		\hline
		
		& &
		\multicolumn{5}{c||}{Phy-HEP} &
		\multicolumn{5}{c|}{Email-Enron}  \\
		
		\cline{3-12}%
		& 
		& \multicolumn{2}{c|}{ {  DCG-SubIneqs} } & \multicolumn{2}{c|}{ {  DCG-SubWarmup} } & \multicolumn{1}{c||}{ {Greedy} } 
		& \multicolumn{2}{c|}{ {  DCG-SubIneqs} } & \multicolumn{2}{c|}{ {  DCG-SubWarmup} } & \multicolumn{1}{|c|}{ {Greedy} }			
		\\

		{$k$} & { $|\Lambda|$}  
		& {Time(s) }& { Cuts(\#) } & {Time(s) }  & {  Cuts(\#) }&  { Time(s)} 		
		& {Time(s) }& { Cuts(\#) } & {Time(s) }  & { Cuts(\#) }&  { Time(s)} 				
		\\
		
		\hline
		
		1 & 100  & 170& 200 & 76&100&616 & 1717& 200 &848&100& 4712\\
		2 & 100  & 650& 208 & 623&200&1141 & 1656 & 200 &2004&200& 7332\\
		3 & 100  & 777& 301 & 771&300&1654 & 1618& 200 &1721&200& 9860\\
		4 & 100  & 1064& 385 & 1054&385&2107 & 1715& 200 &1618&200& 12465\\
		5 & 100  & 375& 491 & 366&491&2530 & 1622& 200 &1628&200& 15137 \\\hline%
		
		1 & 200  & 1800& 400 & 921&200& 906 & 3599& 400 &1810&200&  7610\\
		2 & 200  & 261& 400 & 257&401& 1593 & 4279& 400 &4262&400&  11623\\
		3 & 200  & 524& 599 & 455&598& 2185 & 3372& 400 &3420&400&  15576\\
		4 & 200  & 2208& 770 & 1983&771& 2770 & 3273& 400 &3668&400&  19414\\
		5 & 200  & 845& 989 & 733&988& 3344 & 3265& 400 &3529&400&  23124 \\\hline%
		
		1 & 300  & 416& 600 &245&300&  695 & 5737& 600 &2702&300&  11295\\
		2 & 300  & 414& 609 & 504&600& 1123 & 5256& 600 &5037&600&  17675\\
		3 & 300  & 648& 902 & 567&900& 1561 & 4890& 600 &5003&600&  24357\\
		4 & 300  & 829& 1040 & 863&1038& 2004 & 4921& 600 &6116&726&  31387\\
		5 & 300  & 911& 1190 & 737&1190& 2465 & 4900& 600 &5000&600&  38385 \\ \hline%
		
		1 & 500  & 925 & 1000 & 421 & 500 & 3439 & 9526 & 1000 &	4679 & 500 & 16976 \\					
		2 & 500  & 603 & 1000 & 614 & 1005 & 6250 & 9176 & 1000 & 8756 & 1000 &  25970 \\
		
		3 & 500  & 1266 & 1498 & 1353 & 1500 & 8992  & 9726 & 1000 & 9305 & 1000 & 34775 \\					
		4 & 500  & 1544 & 1985 & 1434 &	1985 & 11545 & 9187 & 1000 &	9507 & 1000 & 43588 \\ 					
		5 & 500  & 2128 & 2220 & 2315 & 2323 & 13808 & 11056 &	1000 & 9390 & 1000 & 52548						
		\\ \hline%
		
		\multicolumn{2}{|c||}{\bf Average} & 917.9& 	839.35&814.6	& 788.75 & 3536.4& 5024.55 & 550 	&4500.15&501.3	& 21190.45 	 \\ \hline

	\end{tabular}
\end{sidewaystable}

From column Cuts(\#) in  Tables \ref{Table:IC-table1}-\ref{Table:IC-table2}, we observe that the number of cuts added to the master problem generally increases with the number of seed nodes $k$.  In other words, more iterations are needed to prove optimality if we have more seed nodes to select. Columns DCG-SubIneqs Time(s) and Cuts(\#) show  that the overall running time does not necessarily increase with the number of user cuts, as more cuts may help the master problem converge to an optimal solution faster. Recall that the running time of  DCG-SubIneqs includes the solution time of the master problem (a mixed-integer program) and the cut generation time of submodular inequalities, which decomposes by each scenario. DCG-SubWarmup  improves the efficiency for most of the instances compared to DCG-SubIneqs, because  it  requires fewer user cuts, and the EmptySetCuts are facet-defining and may lead to faster solution times for the master problem.   This improvement is expected for $k = 1$ (from Proposition \ref{prop:k=1}), but it also holds for $k\ge 2$. 

From columns DCG-SubIneqs Time(s), DCG-SubWarmup Time(s) and Greedy Time(s), we see that the running time of Greedy increases linearly as the number of seed nodes increases, but the same observation can not be made for the number of scenarios. For example, for Phy-HEP with $k=5$ in \tabref{Table:IC-table2}, Greedy takes 2465 seconds to solve the instance with 300 scenarios, but 13808 seconds for the instance with 500 scenarios. In addition, there is no obvious trend in the solution time of DCG-SubIneqs and DCG-SubWarmup  as we increase $k$ or $|\Lambda|$. 

Considering the average solution time, we  observe that  DCG-SubWarmup is  faster than  DCG-SubIneqs, which is in turn faster than Greedy for large networks with more than 10,000 nodes. For example, for  the instance Email-Enron in  \tabref{Table:IC-table2},  the average solution  time  of Greedy  is five times that of  DCG-SubIneqs.  Only for the smallest  instance, UCI-message, Greedy  is the fastest algorithm (see \tabref{Table:IC-table1}). Finally, we note that the time required to calculate $\rho_j^\omega(\bar x)$ for each $\omega\in\Lambda$ and for a given $\bar x\in \mathcal X$ is negligible (it is a reachability problem solved in linear time in $|A_\omega|$). However because this problem needs to be solved for a large number of nodes and scenarios for the calculation of the cut coefficients, the majority of the overall time for DCG is spent on cut generation (for example, the cut generation takes, on average, 80\% of the time over all four problem instances with $|\Lambda|=300$). Because we observe that the major bottleneck is the cut generation time, and that the algorithm converges in a few iterations, we did not implement enhancements  known to improve convergence, such as the trust region method or heuristics \cite[cf.][]{SAGS05,OGH14}.

\exclude{
 This also affects the performance of Greedy, which calculates $\rho_j^\omega(\bar x^i)$, for each $j\in V\setminus R(\bar X^i)$ and $\omega\in \Lambda$ at a solution $x^i$ in iteration $i=1,\dots,k$. So any improvements made in calculating $\rho_j^\omega(\bar x^i)$ will benefit both algorithms.  
}

\subsubsection{Linear Threshold Model} \label{sec:LTexp}

In this section, we summarize our experiments with the linear threshold model. Recall that in the live-arc graph representation of linear threshold models, at most one incoming arc is chosen for each node in the live-arc graph construction  for each scenario. As in \cite{KKT03} we let the deterministic  weight on each arc $(i,j)\in A$ be $w_{ij}=1/\text{indeg}(j)$. We generate $|\Lambda| \in\{ 100, 200, 300, 500\}$  for the  four real-world datasets described earlier.

The results are shown in Tables \ref{Table:LT-table1}-\ref{Table:LT-table2}. Similar to  the independent cascade model, the running time of Greedy   increases linearly in $k$, and there is no obvious trend in the solution time of DCG-SubIneqs and DCG-SubWarmup as we increase $k$ or $|\Lambda|$.   As in the previous experiments for the independent cascade model, DCG-SubWarmup is slower than  Greedy only for the smallest dataset with fewer than 2,000 nodes (UCI-message) (see \tabref{Table:LT-table1}). For the large-scale datasets with over 10,000 nodes (P2P01, Phy-HEP, and Email-Enron) reported in Tables \ref{Table:LT-table1}-\ref{Table:LT-table2}, we observe that the warm-up strategy is highly effective. It provides the best solution times, and fewer iterations and cuts. For example, DCG-SubWarmup outperforms Greedy by a factor of 2.23  in P2P02 in \tabref{Table:LT-table1}, a factor of 4.05  in Phy-Hep in \tabref{Table:LT-table2}, and a factor of 25  in Email-Enron in \tabref{Table:LT-table2}, the largest dataset considered.


\begin{sidewaystable}[htp]
	\caption{Linear Threshold Model for UCI-message and P2P02} 
	\label{Table:LT-table1}
	\centering 
	\begin{tabular}{ |p{1cm}p{1cm}||p{1.5cm}p{1.5cm}|p{1.5cm}p{1.5cm}|p{1.5cm}||p{1.5cm}p{1.5cm}|p{1.5cm}p{1.5cm}|p{1.5cm}|  }
		\hline
		
		& &
		\multicolumn{5}{c||}{UCI-message} &
		\multicolumn{5}{c|}{P2P02}  \\
		
		\cline{3-12}%
		& 
		& \multicolumn{2}{c|}{ {  DCG-SubIneqs} } & \multicolumn{2}{c|}{ {  DCG-SubWarmup} } & \multicolumn{1}{c||}{ {Greedy} } 
		& \multicolumn{2}{c|}{ {  DCG-SubIneqs} } & \multicolumn{2}{c|}{ {  DCG-SubWarmup} } & \multicolumn{1}{|c|}{ {Greedy} }			
		\\

		{$k$} & { $|\Lambda|$}  
		& {Time(s) }& { Cuts(\#) } & {Time(s) }  & {  Cuts(\#) }&  { Time(s)} 		
		& {Time(s) }& { Cuts(\#) } & {Time(s) }  & { Cuts(\#) }&  { Time(s)} 				
		\\
		
		\hline
		
		1 & 100  & 52& 196 & 32&100&1 & 55& 200 & 17&100&110\\
		2 & 100  & 72& 299 & 43&120&1 & 48& 200 & 26&103&222\\
		3 & 100  & 94& 378 & 70&169&2 & 55& 200 & 28&108&337\\
		4 & 100  & 130& 455 & 93&257&3 & 136& 338 & 28&120&449\\
		5 & 100  & 162& 606 & 89&303&4  & 230& 580 & 31&133&565\\\hline%
		
		1 & 200  & 137& 386 &  61&200&1& 394& 400 &  236&200&107 \\
		2 & 200  & 220& 722 &  91&259&3 & 387& 400 &  202&202&211\\
		3 & 200  & 221& 735 &  80&296&4 & 583& 596 &  225&214&311\\
		4 & 200  & 243& 892 &  161&502&6 & 590& 596 &  204&213&414\\
		5 & 200  & 469& 1435 & 289&893&7 & 1260& 1156 &  306&290&512 \\\hline%
		
		1 & 300  & 192& 582 &  103&300&3 & 1457& 600 &  558&300&624\\
		2 & 300  & 188& 590 &  200&513&5 & 1564& 600 &  584&308&1231\\
		3 & 300  & 192& 594 &  159&449&8 & 1853& 885 &  684&343&1833\\
		4 & 300  & 573& 1507 &  249&637&11 & 3529& 1687 &  855&410&2531\\
		5 & 300  & 543& 1554 & 387&1011&13 & 9557& 4000 &  1656&765&3118  \\ \hline%
		
		1&	500&	493&	962&	76&	500&	3&	1225&	1000&	544&	500&	512\\
		2&	500&	462&	988&	84&	559&	6&	1721&	1000&	576&	503&	1028\\
		3&	500&	747&	1403&	181&	740&	9&	1224&	1000&	534&	512&	1507\\
		4&	500&	665&	1415&	130&	723&	12&	4799&	3308&	600&	552&	1969\\
		5&	500&	1282&	2909&	462&	1933&	15&	10505&	5861&	1047&	887&	2415
		\\\hline%

		\multicolumn{2}{|c||}{\bf Average}  	& 356.85& 930.4 &152&523.2& 5.85 & 2058.6& 1230.35 &447.05&338.15& 1000.3	 \\ \hline

	\end{tabular}
\end{sidewaystable}

\begin{sidewaystable}[htp]
	\caption{Linear Threshold Model for Phy-HEP and Email-Enron} 
	\label{Table:LT-table2}
	\centering 
	\begin{tabular}{ |p{1cm}p{1cm}||p{1.5cm}p{1.5cm}|p{1.5cm}p{1.5cm}|p{1.5cm}||p{1.5cm}p{1.5cm}|p{1.5cm}p{1.5cm}|p{1.5cm}|  }
		\hline
		
		& &
		\multicolumn{5}{c||}{Phy-HEP} &
		\multicolumn{5}{c|}{Email-Enron}  \\
		
		\cline{3-12}%
		& 
		& \multicolumn{2}{c|}{ {  DCG-SubIneqs} } & \multicolumn{2}{c|}{ {  DCG-SubWarmup} } & \multicolumn{1}{c||}{ {Greedy} } 
		& \multicolumn{2}{c|}{ {  DCG-SubIneqs} } & \multicolumn{2}{c|}{ {  DCG-SubWarmup} } & \multicolumn{1}{|c|}{ {Greedy} }			
		\\

		{$k$} & { $|\Lambda|$}  
		& {Time(s) }& { Cuts(\#) } & {Time(s) }  & {  Cuts(\#) }&  { Time(s)} 		
		& {Time(s) }& { Cuts(\#) } & {Time(s) }  & { Cuts(\#) }&  { Time(s)} 				
		\\	
		
		\hline
		
		1 & 100  & 476& 200 & 160&100&587 & 209& 200 & 100&100&2045\\
		2 & 100  & 528& 243 & 175&101&1116 & 206& 200 & 97&100&3972\\
		3 & 100  & 1117& 428 & 333&182&1682 & 196& 200 & 131&106&5814\\
		4 & 100  & 1351& 474 & 383&185&2235 & 414& 379 & 158&137&7605\\
		5 & 100  & 786& 387 & 63&209&2740 & 644& 535 & 299&230&9366 \\\hline%
		
		1 & 200  & 1180& 691 &  261&200&595 & 440& 400 &  200&200&2411\\
		2 & 200  & 1141& 586 &  477&359&1176 & 1240& 400 &  608&200&4732\\
		3 & 200  & 1841& 964 &  540&362&1791 & 419& 400 &  223&215&7001\\
		4 & 200  & 1448& 779 &  602&363&2444 & 642& 593 &  285&265&9192\\
		5 & 200  & 1190& 581 &  722&377&3229 & 902& 760 & 447&378&11330 \\\hline%
		
		1 & 300  & 1212& 898 &  424&300&380 & 773& 600 &  324&300&4822\\
		2 & 300  & 399& 600 &  445&303&692 & 740& 600 &  335&300&9283\\
		3 & 300  & 489& 600 &  516&312&1017 & 666& 600 &  375&326&13602\\
		4 & 300  & 943& 871 &  1089&548&1329& 1082& 886 &  436&365&17833 \\
		5 & 300  & 6425& 3413 & 824&619&1642 & 1975& 1441 & 982&710&22354 \\ \hline%

		1&	500&	546&	1000&	236&	500&	2051&	1243&	1000&	670&	500&	6179\\
		2&	500&	1910&	1218&	1229&	887&	4002&	1199&	1000&	613&	500&	12113\\
		3&	500&	2440&	1421&	1312&	891&	6342&	1192&	1000&	675&	523&	17858\\
		4&	500&	4340&	2283&	1635&	988&	8637&	2446&	1929&	845&	669&	23443\\
		5&	500&	9291&	3845&	2051&	1040&	10993&	3490&	2373&	1044&	844&	28876\\
		\hline%

		\multicolumn{2}{|c||}{\bf Average}  	&1952.65&1074.1 &673.85& 441.3 & 2734 &1005.9&774.8& 442.35 & 348.4 & 10991.55  \\ \hline		
	\end{tabular}
\end{sidewaystable}

Some comments are in order for both the independent cascade and linear threshold models. First, we make some observations on increasing $k$. As can be seen from our experiments, the running time of Greedy increases linearly with $k$. However, the increase in the running time of DCG is nonlinear, as can be expected. Hence, as we increase $k$, we need to set some time limits for both DCG and Greedy (currently, we impose no time limits).  In this case, with DCG, we are still able to obtain an incumbent solution with $k$ seed nodes and an estimate on optimality gap provided by the bound from the DCG master. However, with a time limit, we will have to stop Greedy prematurely, before it identifies all $k$ seed nodes. For example, for the independent cascade model for the  medium-sized instance P2P02, setting a time limit of one hour, for $k=30$ and $|\Lambda|=300$, Greedy stops at time limit with a solution that has $k=4$ seed nodes (see Table \ref{Table:IC-table1}). On the other hand, DCG stops with an incumbent solution that has $k=30$ seed nodes and an  optimality gap of 3.7\%.  For the largest instance Enron, from Table \ref{Table:IC-table2}, we observe that Greedy cannot even find the first seed node in over three hours.  Similarly, as we increase  $|\Lambda|$, the running time of both algorithms increase greatly, and as in the case of increasing $k$ we will have to impose time limits and stop Greedy prematurely to be able to compare the performance of the algorithms.

Our computational experiments demonstrate an overlooked opportunity to use optimization methods to solve stochastic influence maximization problems. For {\it deterministic} submodular maximization  problems,  Greedy has much faster solution times and close to optimal performance  in practice than a delayed constraint generation method based on submodular inequalities \cite[see, for example,][for a deterministic hub location problem]{CF14}. In contrast, for  {\it stochastic} submodular maximization  problems, which are very large scale, Greedy, which guarantees a 0.63-approximate solution,  may be much slower than our proposed method. 

\exclude{
Greedy needs to perform the influence spread function evaluation  for each node for each scenario at each iteration, resulting in an overall running time of $O(nmk|\Lambda|)$. 
However, in our experiments we see that
the delayed constraint generation method generally converges to an optimal solution in fewer than $k$
iterations. Even though each iteration requires the solution of a mixed-integer programming (MIP) master problem, the overall solution time is faster than Greedy, due to the efficiency of the state-of-the-art MIP solvers in solving the master problem. This is one of the major contributions of our paper. } 

Finally, note that, throughout the paper, we present a so-called multicut version of the DCG algorithm and its variants, where we add an optimality cut for each scenario at each iteration. We have also tested a single cut implementation, in which multiple cuts across all scenarios are aggregated into a single cut at each iteration. We observe a degraded performance of the single cut version for our problem instances; therefore, we present our results for the multicut approach. In particular,  for the independent cascade model, the total computational time of the single cut version of DCG-SubWarmup is between  20 to 100\% higher than the multicut version in all four datasets with 300 scenarios for $k>1$. 
 \cite[See page 167 of][for a discussion on the problem-dependent nature of the performance of the single vs.\ multicut approach.]{BL97}

\section{Conclusion} \label{sec:conc}

In this paper, we propose a delayed constraint generation  algorithm  to solve influence maximization problems arising in social networks. We show that exploiting the submodularity of the influence function leads to strong optimality cuts. Furthermore, our computational experiments with large-scale real-world test instances  indicate that the algorithm performs favorably against a popular greedy heuristic for this problem. In most instances, our algorithm finds a solution with provable optimality guarantees more quickly than the greedy heuristic, which can only provide a 0.63 performance guarantee. Our algorithm is applicable to many other variants of  the influence maximization problem for which the influence function is submodular. Furthermore, we generalize the proposed algorithm  to solve any two-stage stochastic program, where the second-stage value function is submodular. 

Our results on optimization-based methods for the fundamental influence maximization problems provide a  foundation to build  algorithms for more advanced models, such as the adaptive model of \cite{SS13},  where a subset of additional seed nodes is selected in the second stage based on the realization of some of the uncertain parameters and the seed nodes selected in the first stage. The  decomposition methods of \cite{Sen10,GKS14} and \cite{ZK14} can be employed in this case to convexify the second stage problems that involve binary decisions. Another possible future research direction is to develop optimization-based methods for the problem of marketing to nodes \cite[]{KKT03,KKT15} to increase their probabilities of getting activated.  

\section*{Acknowledgments}
This work is supported, in part, by the National Science
Foundation Grant 1055668.

\appendix

\section{Alternative Benders Optimality Cuts  for Live-Arc Graph Models} \label{sec:app}

In this section, we present optimality cuts that can be obtained by traditional methods.  First, we give an explicit linear programming (LP) formulation for the subproblems \eqref{model:gensub} used to calculate $\sigma_\omega(x)$ for live-arc graph models such as independent cascade or linear threshold. Observe that the maximum number of nodes reachable from nodes $X$ (corresponding to the decision vector $x$) in graph $G_\omega$ can be formulated as a maximum flow problem an a modified graph $G'_\omega=(V\cup\{s,t\}, A'_\omega)$, where $s$ is the source node, $t$ is the sink node, and $A'_\omega$ includes the arcs $A_\omega$ and arcs $(s,i)$ and $(i,t)$ for all $i\in V$. Let the capacity of the arcs $(i,t)$, $i\in V$ be one, and the capacity of arcs $(i,j)\in A_\omega$ be $n$ (the maximum flow possible on any arc). In addition, we would like the arcs $(s,i), i\in V$ to have a capacity of $n$ if $x_i=1$ and 0 otherwise. Therefore, we let the capacity of arc $(s,i)$ be $nx_i$. The reader might wonder why we  create an arc $(s,i)$ if a node $i$ is not activated. To see why, note that in a two-stage stochastic programming framework, we need to build a second-stage model that is correct for any first-stage decision $x$. It is easy to see that the maximum flow on this  graph is equal to the maximum number of vertices reachable from the seeded nodes $X$. 
The LP formulation of the second-stage problem for scenario $\omega\in \Lambda$ is
\begin{subequations}\label{subproblem}
\begin{align}
\sigma_\omega (x)=~~\max \quad &\sum_{i \in V} y_{si} \\
\quad \text{s.t.}~~& \sum_{j: (j,i) \in A'_\omega} y_{ji} - \sum_{j: (i,j) \in A'_\omega} y_{ij} =0,\qquad  i\in V & (u^\omega_i)\\
& y_{si}\le nx_i, \qquad i\in V & (v^\omega_{si}) \\ 
& y_{ij}\le n, \qquad (i,j)\in A_\omega & (v^\omega_{ij}) \\ 
& y_{it}\le 1, \qquad i\in V & (v^\omega_{it}) \\ 
& y_{ij} \ge 0,\qquad  (i,j)\in A'_\omega,
\end{align}
\end{subequations}
where $y_{ij}$ represents the flow on arc $(i,j)\in A'_\omega$, and the dual variables associated with each constraint are defined in parentheses. Note that the subproblems are feasible for any $\omega\in \Lambda$ and $x\in \{0,1\}^n$ (we can always send zero flows), therefore this problem is said to have  {\it complete recourse}.   The dual of the second-stage problem \eqref{subproblem} is 
\begin{subequations}\label{dual-sub}
\begin{align}
\sigma_\omega (x)=~~\min \quad &\sum_{i \in V} (nx_i v^\omega_{si} +v^\omega_{it}) + \sum_{(i,j) \in A_\omega} nv^\omega_{ij}  \\
\quad \text{s.t.}~~& u^\omega_i+ v^\omega_{si}\ge 1, \qquad i\in V & \\ 
& u^\omega_j-u^\omega_i+ v^\omega_{ij}\ge 0, \qquad (i,j)\in A_\omega & \\ 
& -u^\omega_i+v^\omega_{it}\ge 0, \qquad i\in V &  \\ 
& v^\omega_{ij} \ge 0,\qquad  (i,j)\in A'_\omega.
\end{align}
\end{subequations}

Note that we can write a large-scale mixed-integer program, known as the deterministic equivalent program (DEP), to solve the independent cascade problem. To do this, we create copies of the second-stage variables  $y_{ij}^\omega$ for all $\omega\in \Lambda$, where $y_{ij}^\omega$ represents the flow on arc $(i,j)\in A'_\omega$ under scenario $\omega\in \Lambda$. The DEP is formulated as
\begin{subequations}\label{dep}
\begin{align}
\max~~& \sum_{\omega\in \Lambda} p_\omega \sum_{i \in V} y_{si}^\omega\\
\text{s.t.}~~&\sum_{j\in V} x_j \le k  \\
& \sum_{j: (j,i) \in A'_\omega} y_{ji} - \sum_{j: (i,j) \in A'_\omega} y_{ij} =0,\qquad  i\in V & \\
& y_{si}^\omega\le nx_i, \qquad \omega\in \Lambda, i\in V &  \\ 
& y_{ij}^\omega\le n, \qquad  \omega\in \Lambda, (i,j)\in A_\omega &  \\ 
& y_{it}^\omega\le 1, \qquad  \omega\in \Lambda, i\in V & \\ 
& x\in\{0,1\}^n, y_{ij}^\omega\ge 0, \omega\in \Lambda, (i,j)\in A'_\omega.
\end{align}
\end{subequations}
It is well-established in the stochastic programming field that due to its large size, it is not practical to solve DEP directly. Instead, as is commonly done, we consider the use of Benders decomposition method  \cite[]{B62,VW69} utilizing the structure of this large-scale MIP.

A naive way of generating the optimality cuts is to solve the subproblem \eqref{subproblem} for each $\omega\in \Lambda$  as an LP   (in line \ref{alg:step-solve} of Algorithm \ref{alg:benders})
 to obtain $\sigma_\omega (\bar x)$, and  the corresponding dual vector $(\bar u^\omega,\bar v^\omega)$. Then the optimality cut is 
\begin{equation} 
\label{eq:optcut-lp}
\theta_\omega\le \sum_{i \in V} (nx_i \bar v^\omega_{si} +\bar v^\omega_{it}) + \sum_{(i,j) \in A_\omega} n\bar v^\omega_{ij}. 
\end{equation}
We refer to the optimality cuts \eqref{eq:optcut-lp} obtained by solving the subproblems as an LP as the {\it LP-based optimality cuts}.

Next, we discuss a more efficient way of obtaining the optimality cuts   by utilizing the fact that the subproblems are maximum flow problems, which can be solved in polynomial time using specialized algorithms. 
In particular, for our problem, one only needs to  solve a reachability problem to obtain the corresponding maximum flow. Reachability problem in a graph can be solved in linear time in the number of arcs using breadth- or depth-first search. We describe the equivalence of the maximum flow problem defining the evaluation of the influence spread to the graph reachability problem next. 

For a given  first-stage solution $\bar x$ and the corresponding seed set $\bar X$, let $\hat R(\bar X)\subseteq V$ be the set of nodes in $V$ reachable from $s$, $R(\bar X)=\hat R(\bar X)\setminus \bar X$ be the set of nodes reachable from $s$ not including the seed nodes $\bar X$, and $\bar R(\bar X)=V\setminus \hat R(\bar X)$ be the set of nodes in $V$ {\it not} reachable from $s$ in $G'_\omega$.  From maximum  flow minimum cut theorem \cite[see, e.g.,][]{AMO93}) we can show that a minimum cut is given by $(\hat R(\bar X)\cup\{s\}, \bar R(\bar X) \cup\{t\})$.  
(See the maximum flow formulation of this problem for a given $\bar X$ and scenario $\omega\in \Lambda$  in \figref{fig:mincut}.)  
Let $u^\omega_i=1$ if $i\in \hat R(\bar X)$, and $u^\omega_i=0$, if $i\in \bar R(\bar X)$. In addition, for $(i,j)\in A_\omega'$, let $v^\omega_{ij}=1$ if $i\in \hat R(\bar X)\cup\{s\}$ and $j\in \bar R(\bar X) \cup \{t\}$, otherwise let $v^\omega_{ij}=0$. It is easy to check that this choice of the dual variables is feasible. Furthermore, this choice is optimal. To see this, note that the objective value of the dual is 
\begin{align*}
\sum_{i \in V} (nx_i \bar v^\omega_{si} +\bar v^\omega_{it}) + \sum_{(i,j) \in A_\omega} n\bar v^\omega_{ij} &= \sum_{i \in \bar R(\bar X)} nx_i  +  \sum_{i \in  \hat R(\bar X)} 1 + \sum_{(i,j) \in (\hat R(\bar X),\bar R(\bar X)) } n = |\hat R(\bar X)|,
\end{align*}
because $x_i=0$ for $i\in \bar R(\bar X)$ and there can be no arc $(i,j)\in A_\omega$ with $i\in \hat R(\bar X)$, $j\in \bar R(\bar X)$ (otherwise $j$ would be reachable from $s$ and hence it will be in $\hat R(\bar X)$). Because the optimal objective value of the primal subproblem is $\sigma_\omega (\bar x)=|\hat R(\bar X)|$, this dual solution must be optimal. With this choice of the  optimal dual vector, we obtain the Benders optimality cut
\begin{equation} 
\label{eq:optcut-comb}
\theta_\omega \le \sigma_\omega(\bar x) +  \sum_{i \in \bar R(\bar X)} nx_i. 
\end{equation}
We refer to the optimality cuts \eqref{eq:optcut-comb} obtained by solving the subproblems as   reachability problems as  {\it combinatorial optimality cuts}. \cite{W87} and \cite{W91} use the same type of optimality cuts for a problem of investing in arc capacities of a network to maximize flow under stochastic demands. This problem is a relaxation of our problem in that the first stage variables are continuous. Hence, submodular inequalities cannot be used in their problem context.

Note that inequality \eqref{eq:optcut-comb} can also be seen as a big-M type inequality. For $x=\bar x$, with the associated seed set $\bar X$, we get a correct upper bound on $\theta_\omega$ as $\sigma_\omega(x)$. For any other $x\ne \bar x$, if $x_i=1$ for some $i\in \bar R(\bar X)$, then the upper bound on $\theta_\omega$ given by inequality  \eqref{eq:optcut-comb} is trivially valid, because  $\sigma_\omega(x)\le n$ for any $x\in \{0,1\}^n$. Finally, for any $x\ne \bar x$, if $x_i=0$ for all $i\in \bar R(\bar X)$, then we must have $x_j=0$ for some $j\in \bar X$ and $x_\ell=1$ from some $\ell\in R(\bar X)$. However, because $\ell$ is reachable from $\bar X$, replacing $j$ with $\ell$ will not increase the number of reachable nodes, i.e., $\sigma_\omega(x)\le \sigma_\omega(\bar x)$. Therefore, inequality  \eqref{eq:optcut-comb} is  valid. 

\cite{MW81}  propose a method to strengthen Benders cuts in cases when the dual of the subproblems is  degenerate \cite[see also,][for other enhancements of this method]{P08,SL13}. The method chooses, among alternative dual optimal solutions to the subproblem, one that is not dominated.  While this idea is useful to strengthen the weak Benders cut \eqref{eq:optcut-lp} (in particular, inequality \eqref{eq:optcut-comb} corresponding to one choice of optimal dual solutions), we note that it alone cannot lead to the stronger cuts given by 
the submodular inequalities \eqref{eq:optcut-submod}. To see this note, first, that all extreme points of the dual subproblem \eqref{dual-sub} are integral. So any non-integral dual feasible solution is a convex combination of these extreme points. Then note that an optimality cut of the form \eqref{eq:optcut-lp} obtained from the dual is non-dominated only if the corresponding dual solution is an extreme point (otherwise the optimality cut would be a convex combination of the optimality cuts corresponding to the extreme points). As a result, $ \bar v^\omega_{si}\in \mathbb Z$ for all $i\in V$, hence submodular inequalities \eqref{eq:optcut-submod} cannot be expressed as inequalities \eqref{eq:optcut-lp} obtained from non-dominated  extreme point optimal dual solutions to the subproblem \eqref{dual-sub}. 

\begin{figure}[htb]
\centering
	\includegraphics[width=7.5cm]{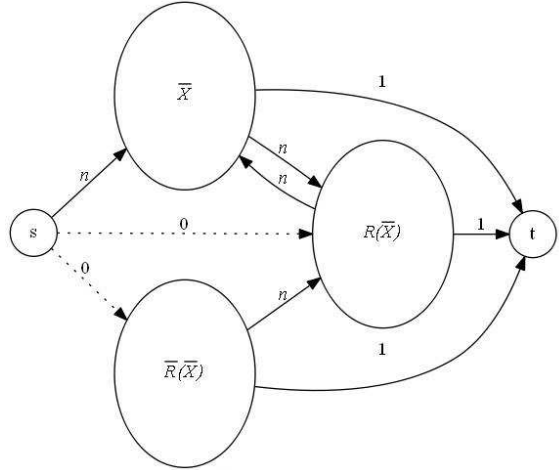}
\caption{Maximum flow  formulation of the influence function.}\label{fig:mincut}
\end{figure}

Finally, note that  because the first-stage problem is a pure binary optimization problem, one can also consider the optimality cuts proposed in the  integer L-shaped method of \cite{LL93}. The resulting inequality, for $\omega\in \Lambda$ and a given $\bar x$, with an associated seed set $\bar X$, is
\begin{equation} \label{eq:optcut-ll}
\theta_\omega \le \sigma_\omega(\bar x) +  \sum_{i \in V\setminus \bar X} (n-\sigma_\omega(\bar x))x_i. 
\end{equation}
This inequality can be strengthened by the same observation  that replacing a node $j\in \bar X$ with  a node $\ell\in R(\bar X)$ does not increase the number of reachable nodes. Therefore, we can reduce the coefficient of $x_\ell$ in inequality \eqref{eq:optcut-ll} to obtain a strengthened version of the integer L-shaped optimality cut \eqref{eq:optcut-ll}: 
\begin{equation} \label{eq:optcut-llst}
\theta_\omega \le \sigma_\omega(\bar x) +  \sum_{i \in \bar R(\bar X)} (n-\sigma_\omega(\bar x))x_i, 
\end{equation}
which is clearly valid.  We refer to inequalities \eqref {eq:optcut-llst} as the {\it strengthened integer L-shaped optimality cuts}.

\begin{proposition}
The submodular optimality cuts \eqref{eq:optcut-submod}  dominate the combinatorial optimality cuts  \eqref{eq:optcut-llst}. 
\end{proposition}

\begin{proof}
This follows because $r_j^\omega(S)\le n-\sigma_\omega(\bar x)$ for any $j\in \bar R(S)$. 
\end{proof}

\subsection{Computations with Benders using strengthened L-shaped cuts} \label{sec:lc-comp}

In our computational study in Section \ref{sec:hepnet}, we set  $\pi_{ij}=p=0.1, (i,j)\in A$ in the real world network. Because the influence probability $p$ is very small, the live-arc graphs corresponding to each scenario are  large-scale sparse networks. We were not able to solve even the smallest instances (with $k=1$ and $|\Lambda|=50$) using Benders-LC  after one day. To demonstrate the inefficiency of Benders-LC,  we consider a much smaller subset of the sparse HEPT network under one scenario,  depicted in \figref{Fig:L} with 15 nodes and 4 directed arcs, and compare the performance of DCG-SubIneqs and Benders-LC. In other words, we let $p=1$, which leads to a deterministic problem (i.e., a unique scenario with objective $\theta_1$ and $p_1=1$). We vary the value of $k$ from 1 to 5. 

\begin{figure}[htbp]
	\centering	
	\includegraphics[width=10.5cm]{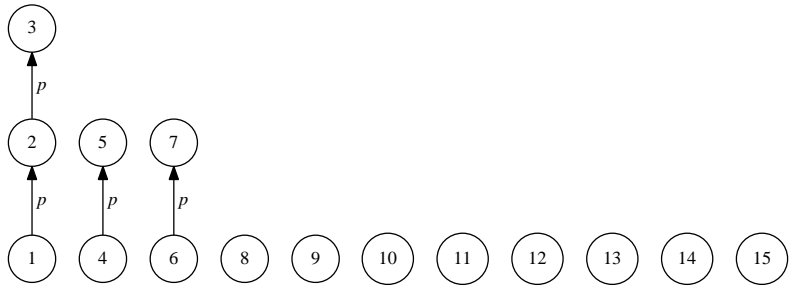}
	\caption{Sparse Network with 15 nodes and 4 arcs with equal  influence probabilities $p$.}
	\label{Fig:L} 
\end{figure}

The total number of user cuts added to the corresponding master problem  is shown in \tabref{Table:L}. We observe that compared to DCG-SubIneqs the number of user cuts added to the master problem of Benders-LC grows rapidly as the number of seed nodes $k$ increases. Indeed,  the number of user cuts for Benders-LC approached  $\binom {15}k$,  indicating that  Benders-LC is effectively a pure enumeration algorithm for this problem. The strengthened integer L-shaped  optimality cuts \eqref{eq:optcut-llst} do not provide any useful information on the objective value when the solution is different from the one that generates the cut. In contrast, submodular inequalities are highly effective for this set of problems.  To see why, consider the problem of finding $k=1$ seed node. The master problem of both DCG-SubIneqs and Benders-LC selects $k=1$ node arbitrarily, because they do not have any cut at the beginning. Because the sparse network is constituted of many singleton nodes (with no incoming and outgoing arcs), there is a high probability  that the master problem selects one singleton at the first iteration. Suppose that the master problem chooses  node 15, which was also the choice of  CPLEX. DCG-SubIneqs generates the cut $$\theta_1 \le 1+ 3x_1 + 2x_2 + x_3 + 2x_4 + x_5 + 2x_6 + x_7 + x_8 + x_9 + x_{10} + x_{11} + x_{12} + x_{13} + x_{14},$$ and Benders-LC generates the cut $$\theta_1 \le 1+ \sum\limits_{i=1}^{14} 14x_i,$$ to be added to the corresponding master problem. At the second iteration, due to the use of the stronger optimality cut, DCG-SubIneqs chooses  node 1 and reaches  optimality, but Benders-LC chooses one of the 14 nodes arbitrarily. Note that, in the worst case, Benders-LC  traces all 15 nodes in the network (and generates 15 optimality cuts) before reaching the optimal solution. Therefore, in the  large-scale network of Section \ref{sec:hepnet}, Benders-LC fails due to the need for a large number of iterations and computational time.  In contrast, the submodular inequality guides the master problem to choose nodes with higher marginal influence.

\begin{table}[htb]
	\caption{Comparison of DCG-SubIneqs and Benders-LC.}
	\label{Table:L}	
	\begin{center}
		\begin{tabular}{|c|c|c|c|c|c|} \hline
						& \multicolumn{5}{|c|}{Number of user cuts with different $k$}  \\
						\cline{2-6} %
			Algorithm & $k=1$ & $k=2$ & $k=3$ & $k=4$ & $k=5$ \\
			\hline
			DCG-SubIneqs & 2 & 5 & 11 & 16 & 51 \\
			Benders-LC		 & 15 & 106 & 458  & 1365 & 3003	 \\\hline%
		\end{tabular}
		
	\end{center}
\end{table}

\bibliographystyle{apalike}
\bibliography{ref_social}

\end{document}